\documentclass[lettersize,journal]{IEEEtran}
\usepackage{amsmath,amsfonts,algorithmic,array,subfig,textcomp,stfloats,url,verbatim,graphicx,algorithm,algorithmic,hyperref,float,cite,balance,xcolor}

\newcommand{\ym}[1]{\textcolor{black}{#1}}

\def\BibTeX{{\rm B\kern-.05em{\sc i\kern-.025em b}\kern-.08em
		T\kern-.1667em\lower.7ex\hbox{E}\kern-.125emX}}
\newtheorem{theorem}{Theorem}[section]

\newtheorem{proposition}[theorem]{Proposition}
\newtheorem{remark}[theorem]{Remark}

\newenvironment{proof}[1][Proof]{\noindent\textbf{#1.} }{\ \rule{0.5em}{0.5em}}

\begin{document}
	
	\title{Distributed ADMM Approach for the Power Distribution Network Reconfiguration}
	
	\author{Yacine Mokhtari\thanks{LIAS (UR 20299), ISAE-ENSMA/Universit\'{e} de Poitiers, 2 rue Pierre Brousse, 86073 Poitiers Cedex 9, France. E-mail: yacine.mokhtari@ensma.fr, patrick.coirault@univ-poitiers.fr}, Patrick Coirault, Emmanuel Moulay\thanks{XLIM (UMR CNRS 7252), Universit\'{e} de Poitiers, 11 bd Marie et Pierre Curie, 86073 Poitiers Cedex 9, France. E-mail: emmanuel.moulay@univ-poitiers.fr}, Jerome Le Ny\thanks{Department of Electrical Engineering, Polytechnique Montr\'eal and GERAD, Montreal, QC H3T-1J4, Canada. E-mail: jerome.le-ny@polymtl.ca} and Didier Larraillet\thanks{SRD Energies, 78 Av. Jacques Coeur, 86000 Poitiers. E-mail: didier.larraillet@srd-energies.fr}}
	
	\date{}
	\maketitle
	
	\begin{abstract}
		The electrical network reconfiguration problem aims to minimize losses in a distribution system by adjusting switches while ensuring radial topology. The growing use of renewable energy and the complexity of managing modern power grids make solving the reconfiguration problem crucial. Distributed algorithms help optimize grid configurations, ensuring efficient adaptation to changing conditions and better utilization of renewable energy sources. This paper introduces a distributed algorithm designed to tackle the problem of power distribution network reconfiguration with a radiality constraint. This algorithm relies on ADMM (Alternating Direction Method of Multipliers), where each agent progressively updates its estimation based on the information exchanged with neighboring agents. We show that every agent is required to solve a linearly constrained convex quadratic programming problem and a Minimum Weight Rooted Arborescence Problem (MWRAP) with local weights during each iteration. Through numerical experiments, we demonstrate the performance of the proposed algorithm in various scenarios, including its application to a 33-bus test system and a real-world network.
	\end{abstract}
	
	\begin{IEEEkeywords}
		ADMM, minimum spanning tree, arborescence problem, graph theory, mixed-integer non-linear programming (MINLP), optimal network configuration, power loss minimization, radial distribution networks, distributed optimization.
	\end{IEEEkeywords}

	\section{Introduction}
	
	The Power Distribution Network Reconfiguration (PDNR) problem involves strategically modifying the network's topology by opening and closing switches. The primary aim is to optimize system performance by minimizing power losses while ensuring the network maintains its radial structure. Radial distribution systems are preferred due to their lower maintenance requirements and fewer potential failure points compared to alternative network topologies, as indicated in \cite{saboori2015reliability}.
	
	PDNR methods can be categorized into heuristic \cite{merlin1975search,civanlar1988distribution,baran1989network,ababei2010efficient}, meta-heuristic \cite{jeon2002efficient,wang2008optimization,rao2010optimal} and mathematical optimization-based approaches \cite{khodr2009distribution,taylor2012convex,jabr2012minimum,montoya2020mixed,mokhtari2025alternating}. Various methods were employed to address the PDNR problem using centralized methods \cite{mishra2017comprehensive}. These methods gather all network data and transmit it back to the control center to determine the reconfiguration solution. However, this approach is vulnerable to single-point failures and cyber-attacks. Furthermore, it can be computationally expensive and unsuitable for large-scale systems.
	
	\ym{In response to the limitations of centralized reconfiguration, recent efforts have focused on distributed (or decentralized) algorithms, which offer enhanced scalability, improved fault tolerance, and better adaptability to large-scale systems. In such frameworks, computing agents exchange only limited information with a selected subset of neighboring agents, reducing communication overhead and preserving data privacy. This decentralized structure not only enhances cybersecurity but also lowers the infrastructure costs associated with global coordination. Furthermore, distributed algorithms are inherently robust to individual agent failures, support the integration of renewable energy sources, and enable parallel computation, thereby offering faster convergence and improved scalability compared to centralized approaches~\cite{molzahn2017survey}}.
	
	\ym{In the context of the distributed approach for the PDNR problem, the multi-agent approach has been explored in several studies, including~\cite{nagata2002multi,elmitwally2015fuzzy,li2019full}. However, these approaches are not fully distributed, as they require some or all agents to access global network information, which raises concerns regarding data privacy and introduces vulnerability to single-point failures. Additionally, they often fail to address the need for scalable and efficient solutions suitable for large-scale distribution networks. An alternative method for solving the distributed PDNR problem is the Alternating Direction Method of Multipliers (ADMM), as proposed in~\cite{shen2019distributed,nejad2021enhancing,lopez2023enhanced}}.
	
	The ADMM algorithm is a primal-dual splitting method initially designed to solve convex optimization problems \cite{Boyd1,eckstein2015understanding}. ADMM as a general heuristic for solving non-convex problems is mentioned in \cite[Ch 9]{Boyd1} and this idea is further explored in \cite{boyd-biconvex,takapoui2020simple}. Recently, it has garnered attention as a method for discovering approximate solutions for NP-hard problems. Utilizing ADMM for mixed-integer problems offers a significant benefit by requiring only a method to compute the projection onto the discrete constraint set, either exactly or approximately if the computation is costly \cite{boyd-biconvex}. While there is no guarantee that ADMM-based methods converge to a global solution for non-convex problems, they are still useful in finding high-quality local solutions.
	
	\ym{Regarding the application of ADMM to PDNR problems, several distributed ADMM-based algorithms have been proposed in the literature~\cite{shen2019distributed,nejad2021enhancing,lopez2023enhanced}. In~\cite{shen2019distributed}, the authors introduced a distributed ADMM algorithm specifically tailored to the PDNR problem. Due to the inherent difficulty of directly handling binary decision variables within the ADMM framework, a relaxation strategy was adopted, as introduced in~\cite{boyd-biconvex}. In this work, the relaxation approach is presented as a general strategy for handling non-convex constraints within ADMM. However, when applied to PDNR, this method may lead to non-radial network topologies, thus rendering the obtained solution infeasible. The core issue arises in the projection step, which relies on a naive rounding of continuous values to binary decisions (0 or 1)~\cite{boyd-biconvex}, without explicitly enforcing the radiality constraint.}

	\ym{To address this issue, an enhanced ADMM-based algorithm was proposed in \cite{nejad2021enhancing}. This algorithm incorporates a proximal operator with a penalty parameter based on residual errors to handle binary variables. Nevertheless, there is still no guarantee of achieving a radial solution since the projection step remains on the binary set. In \cite{lopez2023enhanced}, additional conditions were introduced in the projection step to ensure radiality. The authors combined the algorithms from \cite{shen2019distributed} and \cite{nejad2021enhancing} with their own, selecting the optimal (feasible) solution among the three. Notably, the examination of radiality constraints on binary variables occurs only after the algorithm's convergence.}
	
	\ym{Recently, a centralized ADMM-based algorithm was proposed in~\cite{mokhtari2025alternating} to solve the PDNR problem. The projection subproblem is addressed exactly using two approaches: the relaxation technique, as used in prior works~\cite{shen2019distributed,nejad2021enhancing,lopez2023enhanced}, and a novel variable substitution technique introduced in~\cite{mokhtari2025alternating}. Both methods successfully solve the projection step by reformulating it as a Minimum Weight Rooted Arborescence Problem (MWRAP) \cite{korte2018combinatorial}. However, the substitution technique demonstrated superior performance and favorable properties compared to relaxation. Despite these advances, the proposed centralized algorithm has not yet been extended to the distributed setting.}
	
	\emph{Contributions:} \ym{This paper introduces a novel distributed algorithm for solving the PDNR problem, with a particular emphasis on preserving radiality throughout the optimization process. It extends the centralized framework proposed in~\cite{mokhtari2025alternating} to a fully decentralized setting, where agents interact solely through local communication with their neighbors. The algorithm is built upon the ADMM framework and leverages a variable substitution technique, initially introduced in~\cite{mokhtari2025alternating}, which demonstrates superior performance compared to the relaxation-based strategies commonly used in earlier works~\cite{shen2019distributed,nejad2021enhancing,lopez2023enhanced}. Radiality is rigorously enforced at each iteration by reformulating the projection step as a MWRAP with local weights, thereby ensuring the feasibility of the resulting network configuration. Furthermore, we introduce an enhanced version of the relaxation-based algorithms proposed in~\cite{shen2019distributed,nejad2021enhancing,lopez2023enhanced}, in which the projection step is solved exactly via a MWRAP, ensuring radiality at each iteration. A comparative analysis on the standard 33-bus distribution test system demonstrates that the proposed variable substitution method consistently outperforms the enhanced relaxation-based approach. This performance advantage, previously established in the centralized setting~\cite{mokhtari2025alternating}, is now shown to extend to the distributed framework as well.}

	The article is structured as follows. The proposed mathematical model is presented in Section~\ref{Sec:mod}. Section~\ref{Sec:algo} describes the optimization algorithm, and Section~\ref{Sec:num} presents numerical simulation results.
	
	\section{The Optimization Problem for the Simplified DistFlow Model}\label{Sec:mod}
	
	The structure of the power distribution network can be depicted as a graph containing nodes and connections known as edges (or arcs if a direction is considered). Within this graph, the edges symbolize the electrical lines responsible for carrying power, and the resistivity properties of these lines result in electrical losses, contributing to a specific cost per line. The grid consists of a generator, which holds all the power; source nodes, also known as substations or feeders, which do not hold any power but simply transmit energy to consumer nodes or buses; renewable energy nodes capable of both producing and consuming power; and load buses dedicated solely to power consumption. It is essential to highlight that substations do not receive power from any other nodes except the generator. In our context, the nodes can be simple buses or clusters. We assume that every node in a grid is equipped with a processor capable of making computations as an agent. We also assume that all the lines are switchable.
	
	Let us now introduce some notations that will be useful later. Consider a finite directed strongly connected graph $G(\mathcal{V},\mathcal{A})$, where $\mathcal{V}$ is the set of nodes and $\mathcal{A}\subset \mathcal{V}\times \mathcal{V}$ is the set of arcs. $|\mathcal{V}|$ and $|\mathcal{A}|$ denote the cardinal of $\mathcal{V}$ and $\mathcal{A}$ respectively. An arc going from node $i$ to node $j$ is denoted by the ordered pair $(i,j)$. We assume that each edge in $G$ is bidirectional, i.e., $(j,i) \in \mathcal{A}$ if $(i,j) \in \mathcal{A}$. The generator node is indexed by $g$. Let us denote by $\mathcal{R}\subset \mathcal{V}$ the set of source nodes, assumed connected to $g$, and by $\mathcal{B}:=\mathcal{V}\backslash \left( \mathcal{R}\cup \{g\}\right) $ the set of buses. The Hadamard product, also known as the element-wise product, of the vectors $\boldsymbol{x}\in\mathbb{R}^n$ and $\boldsymbol{y}\in\mathbb{R}^n$ is defined by:
	\[
	\boldsymbol{x}\odot \boldsymbol{y}:=\left(x_1y_1,\ldots,x_ny_n\right)^T.
	\]
	Bold symbols consistently represent vectors or matrices. If $\boldsymbol{x},\boldsymbol{y},\boldsymbol{z}\in\mathbb{R}^n$ then $\boldsymbol{z}\in[\boldsymbol{x}, \boldsymbol{y}]$ means that $z_i\in\left[x_i, y_i\right]$ for all $i\in\{1,\ldots,n\}$. Moreover, $[\boldsymbol{x}]_{e}$ denotes the $e$-th element of the vector $\boldsymbol{x}$. The vectors $\boldsymbol{0}_n$ and $\boldsymbol{1}_n$ denote vectors of size $n$ consisting entirely of zeros and ones, respectively.
	
	Here, we model EDS by using the \emph{simplified DistFlow model} (SDM) introduced in \cite{baran1989network}, see also \cite[Section III.A]{taylor2012convex} for more details. The formula expressing power loss on the arc connecting nodes $i$ and $j$ can be stated as follows \cite{baran1989network}:
	\begin{equation}
		\text{loss}_{ij}=\frac{r_{ij}\left( P_{ij}^{2}+Q_{ij}^{2}\right) }{V_{i}^{2}}. \label{loss}
	\end{equation}
	In this equation, $r_{ij}>0$, $P_{ij}\geq0$ and $Q_{ij}\geq0$ correspond to the resistance, active, and reactive power flow through the connection $(i,j)$. $V_{i}$ denotes the voltage magnitude at the node $i$. We assume that the voltage variations along the network must stay within strict regulatory bounds:
	\begin{equation}
		\left(1-\epsilon\right)^2 \leq \frac{V_i^2}{V_0^2} \leq \left(1+\epsilon\right)^2,\quad \forall i\in\mathcal{V},\label{interval}
	\end{equation}
	where $V_{0}$ represents the reference base voltage. In general, we have $\epsilon \approx 0.05$, see for instance \cite{baran1989network}.
	
	According to \cite{baran1989network,taylor2012convex}, for each $(i, j) \in \mathcal{A}$, the following relation between $V_{i}^2$ and $V_{j}^2$ holds:
	\begin{equation}
		V_{j}^2=V_{i}^2-2\left(r_{ij}P_{ij}+x_{ij}Q_{ij}\right)+\left(r_{ij}^2+x_{ij}^2 \right)\frac{P_{ij}^2+Q_{ij}^2}{V_{i}^2}, \label{voltage}
	\end{equation}
	where $x_{ij}>0$ is the reactance of the line $(i,j)$. We denote:
	\[
	U_i:=\frac{V_i^2}{V_0^2}, \quad i\in \mathcal{V},
	\]
	and for the sake of clarity we keep the notation ${r_{ij}}$ and ${x_{ij}}$ for $\frac{r_{ij}}{V_0^2}$ and $\frac{x_{ij}}{V_0^2}$, respectively. In addition to the above constraints, the forward updates in the buses must be satisfied for each bus $i\in \mathcal{V}$:
	\begin{equation}\label{div1}
		\begin{array}{lll}
			\operatorname{div}(\boldsymbol{P})_{i} &:= &\sum\limits_{j:(i,j)\in \mathcal{A}} P_{ij} - \sum\limits_{j:(j,i)\in \mathcal{A}} P_{ji} \\
			&=& \rho _{i}^{1} - \sum\limits_{j:(i,j)\in \mathcal{A}} r_{ij} \frac{P_{ij}^{2} + Q_{ij}^{2}}{V_{i}^{2}},
		\end{array}
	\end{equation}
	and
	\begin{equation}
		\begin{array}{lll}
			\operatorname{div}(\boldsymbol{Q})_{i} &:= &\sum\limits_{j:(i,j)\in \mathcal{A}} Q_{ij} - \sum\limits_{j:(j,i)\in \mathcal{A}} Q_{ji} \label{div2} \\
			&=& \rho_{i}^{2} - \sum\limits_{j:(i,j)\in \mathcal{A}} x_{ij} \frac{P_{ij}^{2} + Q_{ij}^{2}}{V_{i}^{2}},
		\end{array}
	\end{equation}
	where $\rho_i^1\leq0$ represents the active power and $\rho_i^2\leq0$ represents the reactive power injected at node $i\in \mathcal{B}$, for $k=1,2$. In the following, we assume that the demand matches the production within the network:  
	$$
	\sum_{i \in \mathcal{B}} \rho_i^k = \rho_g^k, \quad k=1,2.
	$$
	Additionally, we assume that the power flows $P_{ij}$ and $Q_{ij}$ should not exceed the capacities $\bar{P}_{ij}$ and $\bar{Q}_{ij}$ of the line, which leads to:	
	$$
	0 \leq P_{ij} \leq \bar{P}_{ij}, \qquad 0 \leq Q_{ij} \leq \bar{Q}_{ij}.
	$$

	As pointed out in \cite{baran1989network}, in the SDM approximation the quadratic terms in \eqref{voltage}, \eqref{div1}, and \eqref{div2} can be dropped since they represent the losses on the arcs, and hence they are much smaller than the branch power terms $\boldsymbol{P}\in\mathbb{R}^{|\mathcal{A}|}$ and $\boldsymbol{Q}\in\mathbb{R}^{|\mathcal{A}|}$. Furthermore, thanks to \eqref{interval} and for the sake of convexifying the problem~\eqref{loss}, the objective can be approximated at each arc $(i,j)$ as:
	\begin{equation*}
		\text{loss}_{ij} \approx r_{ij}\left(
		P_{ij}^{2}+Q_{ij}^{2}\right) . \label{loss2}
	\end{equation*} 
	Now, we introduce binary decision variables $b_{ij} \in \{0,1\}$ for each arc $(i,j) \in \mathcal{A}$. Setting $b_{ij} = 1$ indicates power flow from $i$ to $j$, while $b_{ij} = 0$ indicates no power flow. For a line $(i,j)$ to remain open, power must flow in only one direction, so only one of $b_{ij}$ or $b_{ji}$ can be equal to one. Conversely, setting both $b_{ij} = b_{ji} = 0$ effectively closes the line $(i,j)$. These variables $b_{ij}$ are also useful in formulating mathematical constraints that enforce a radial (or arborescence) topology rooted at $g$ for the network, as demonstrated in \cite{taylor2012convex, jabr2012minimum, wang2020radiality, lei2020radiality}. However, the explicit form of these constraints is not provided here, as it will not be used in the following discussion.
	
	The reconfiguration problem for the SDM can be written as follows:
	\begin{subequations}\label{DOT}
		\begin{align}
			&\underset{\left( {\boldsymbol{P},\boldsymbol{Q},\boldsymbol{U},\boldsymbol{b}}\right)\in \mathbb{X}\times\mathbb{S} }\min (\boldsymbol{r}\odot\boldsymbol{b})^T(\boldsymbol{P}\odot\boldsymbol{P}+\boldsymbol{Q}\odot\boldsymbol{Q}), \label{objective}\\
			&\text{subject to} \nonumber \\
			&\operatorname{div}(\boldsymbol{b}\odot \boldsymbol{P})= \boldsymbol{\rho }^{1}, \label{eq:constraint1} \\
			&\operatorname{div}(\boldsymbol{b}\odot \boldsymbol{Q})=\boldsymbol{\rho }^{2}, \label{eq:constraint2} \\
			&\boldsymbol{b}\odot \boldsymbol{AU}=2\boldsymbol{b}\odot \left( \boldsymbol{r}\odot \boldsymbol{P}+\boldsymbol{x}\odot \boldsymbol{Q}\right), \label{eq:constraint3}
		\end{align}
	\end{subequations}
	where the sets $\mathbb{X}$ and $\mathbb{S}$ are given by:
	\begin{align}
		& \mathbb{X}=\left[ \boldsymbol{0}_{|\mathcal{A}|},\bar{\boldsymbol{P}}
		\right] \times \left[ \boldsymbol{0}_{|\mathcal{A}|},\bar{\boldsymbol{Q}}
		\right] \times \left[ (1-\epsilon )^{2},(1+\epsilon )^{2}\right] ^{|\mathcal{%
				V}|},  \notag \\
		& \mathbb{S}=\left\{ 
		\begin{array}{c}
			\boldsymbol{b}\in \{0,1\}^{|\mathcal{A}|}:\ \boldsymbol{b}\ \text{forms an
				arborescence} \\ 
			\text{ rooted at }g 
		\end{array}
		\right\} ,  \notag
	\end{align}
	and the matrix $\boldsymbol{A}\in\mathbb{R}^{|\mathcal{A}|\times|\mathcal{V}|}$ is defined such that:
	\begin{equation*}
		[\boldsymbol{A}\boldsymbol{U}]_{e} = U_{i} - U_{j}, \quad \text{if } e = (i,j) \in \mathcal{A}.
	\end{equation*}
	Note that a feasible solution $\boldsymbol{b}$ to the above problem defines an arborescence $G_{\boldsymbol{b}} \subseteq G$ (a directed tree) rooted at $g$, whose arcs are the closed switches. The reader can refer to \cite[Subsection~2.2]{korte2018combinatorial} for a detailed definition of an arborescence. Throughout this paper, we assume the feasibility of problem~\eqref{DOT}, meaning that at least one optimal solution exists.
	
	\begin{remark}
		The voltage drop constraint~\eqref{voltage} has been multiplied component-wise with $\boldsymbol{b}$ in \eqref{DOT}, making it active solely when $b_{ij}=1$. This is different from the “big-$M$” method presented in \cite{taylor2012convex,nejad2019distributed} which includes or excludes these constraints using the inequality:
		\begin{equation*}
			\left\vert V_{j}^{2}-\left(V_{i}^{2}-2\left(r_{ij}P_{ij}+x_{ij}Q_{ij}\right) \right) \right\vert
			\leq (1-b_{ij})M,
		\end{equation*}
		for all $(i,j)\in \mathcal{A}$, where $M$ is a sufficiently large parameter. In our case, we have adopted the formulation \eqref{eq:constraint3} to represent the combinatorial problem as an MWRAP as it will be explained later. 
	\end{remark}
	
	\subsection*{Distribution of computation tasks among agents}
	
	First, we introduce the auxiliary variable $\boldsymbol{Y}$ and $\boldsymbol{Z}$ such that  
	\begin{equation}
		\boldsymbol{Y}=\boldsymbol{P}\odot \boldsymbol{b},\text{ \ }\boldsymbol{Z}= 
		\boldsymbol{Q}\odot \boldsymbol{b}.\label{substitution}
	\end{equation}
	The next step is to associate every agent $i\in \mathcal{V}$ with its objective function $f_{i}$ to minimize. We introduce the local variables $\left\{ \boldsymbol{X}^{i}\right\} _{i\in \mathcal{V}}$ defined by  
	\begin{equation}
		\boldsymbol{X}^{i}=\left( \boldsymbol{Y}^{i},\boldsymbol{Z}^{i},\boldsymbol{P}^{i},\boldsymbol{Q}^{i},\boldsymbol{U}^{i}\right) \in \mathbb{R}^{4| 
			\mathcal{A}|+|V|},  \label{X}
	\end{equation}
	such that 
	\begin{equation}
		\boldsymbol{X}^{i}=\boldsymbol{X}^{j},\quad\forall \left( i,j\right) \in 
		\mathcal{A}.  \label{equality}
	\end{equation}
	Constraint~\eqref{equality} ensures that all agents reach a consensus, meaning they converge to a common limit. We use the notation~\eqref{X} for the sake of brevity in presenting mathematical expressions and it can be omitted as needed. We associate for each $\boldsymbol{X}^{i}$, $i\in \mathcal{V}$, a local binary vector $\boldsymbol{b}^{i}$ such that  
	\begin{eqnarray}
		\boldsymbol{Y}^{i} &=&\boldsymbol{P}^{i}\odot \boldsymbol{b}^{i}, \label{C1} \\
		\boldsymbol{Z}^{i} &=&\boldsymbol{Q}^{i}\odot \boldsymbol{b}^{i}, \label{C2}
		\\
		\boldsymbol{b}^{i}\odot\boldsymbol{A}^{i}\boldsymbol{U}^{i} &=&2\left( 
		\boldsymbol{r}^i\odot\boldsymbol{P}^{i}+\boldsymbol{x}^i \odot \boldsymbol{Z}
		^{i}\right) ,  \label{C3}
	\end{eqnarray}
	where $\boldsymbol{r}^i,\boldsymbol{x}^i\in \mathbb{R}^{|\mathcal{A}|}$, are defined by
	\begin{equation*}
		\left[ \boldsymbol{r}^{i}\right] _{e}=\left\{ 
		\begin{array}{ll}
			\left[ \boldsymbol{r}\right] _{e} & \text{if }e=(i,j)\text{ \textrm{or} }e=(j,i), \\ 
			0, & \mathrm{otherwise},
		\end{array}
		\right. 
	\end{equation*}
	\begin{equation*}
		\left[ \boldsymbol{x}^{i}\right] _{e}=\left\{ 
		\begin{array}{ll}
			\left[ \boldsymbol{x}\right] _{e} & \text{if } e=(i,j)\text{ \textrm{or} }e=(j,i), \\ 
			0, & \mathrm{otherwise},%
		\end{array}%
		\right. 
	\end{equation*}
	and $\boldsymbol{A}^i\in\mathbb{R}^{|\mathcal{A}|\times|\mathcal{V}|}$, $i\in \mathcal{V}$, is defined such that
	\begin{equation*}
		\left[ \boldsymbol{A}^{i}\boldsymbol{U}^{i}\right] _{e}=\left\{ 
		\begin{array}{ll}
			U_{i}^{i}-U_{j}^{i}, & \text{if } e=(i,j), \\ 
			U_{j}^{i}-U_{i}^{i}, & \text{if } e=(j,i), \\ 
			0, & \mathrm{otherwise}.
		\end{array}
		\right. 
	\end{equation*}
	This setup ensures that each agent $i \in \mathcal{V}$ has access only to the components of $\boldsymbol{r}$ and $\boldsymbol{x}$ related to its neighbor. Additionally, we write the constraints~\eqref{eq:constraint2} and \eqref{eq:constraint3} in terms of their respective local variables for every agent  $i\in \mathcal{V}$ as  
	\begin{eqnarray}
		\operatorname{div}(\boldsymbol{Y}^{i})_{i} &=&\rho _{i}^{1},  \label{C5} \\
		\operatorname{div}(\boldsymbol{Z}^{i})_{i} &=&\rho _{i}^{2},  \label{C6}
	\end{eqnarray}
	which is equivalent to the formulation $\operatorname{div}(\boldsymbol{Y})=\boldsymbol{\rho }^{1}$ and $\operatorname{div}(\boldsymbol{Z})=\boldsymbol{\rho}^{2}$ thanks to the connectedness of the graph. 
	
	To distribute the objective function~\eqref{objective} among the agents, let us consider:
	\begin{eqnarray*}
		\sum_{(i,j)\in\mathcal{A}}\left(Y_{ij}^{2}+Z_{ij}^{2}\right) r_{ij}=\sum_{i\in \mathcal{V}}f_{i}(\boldsymbol{X}),
	\end{eqnarray*}
	where $f_{i}(\boldsymbol{X})$ is given by  
	\begin{equation*}
		f_{i}(\boldsymbol{X})=\sum_{j:(i,j)\in\mathcal{A}}\left(Y_{ij}^{2}+Z_{ij}^{2}\right) r_{ij}.
	\end{equation*}
	The distributed version of problem~(\ref{DOT}) writes  
	\begin{equation}
		\begin{array}{ll}
			\underset{\boldsymbol{X}^{i}\in\mathbb{X},\,\boldsymbol{b}^{i}\in 
				\mathbb{S}, \ \forall i \in \mathcal{V}}{\min}  \sum_{i\in
				\mathcal{V}}f_{i}(\boldsymbol{X} ^{i}), \\ 
			\text{\textrm{subject to}} \\
			\boldsymbol{X}^{i}\text{ satisfies \eqref{equality}--\eqref{C6},} \quad\forall i\in \mathcal{V}.
		\end{array}
		\label{consensus}
	\end{equation}
	
	\begin{remark}
		We could opt for a relaxation approach similar to that in \cite{shen2019distributed,nejad2021enhancing,lopez2023enhanced}. However, as it is shown in Section \ref{Sec:num}, employing the variable substitutions \eqref{substitution} results in superior algorithm performance compared to relaxation techniques.
	\end{remark}
	
	\begin{remark}\label{Remark} 
		It is important to note that problem~\eqref{consensus} is a relaxed version of problem~\eqref{DOT}, and they are not equivalent. The constraints \eqref{equality}--\eqref{C3} do not ensure consensus on the binary vectors, i.e., $\boldsymbol{b}^{i} = \boldsymbol{b}^{j}, \ \forall \left( i, j \right) \in \mathcal{A}.$  To show this, let us focus only on \eqref{C1} since the analysis for \eqref{C2} and \eqref{C3} is similar. We have for some $\left( i,j\right) \in \mathcal{A}$ from \eqref{C1}: $\boldsymbol{Y}^{i}=\boldsymbol{P}^{i}\odot \boldsymbol{b}^{i}$ and $\boldsymbol{Y}^{j}=\boldsymbol{P}^{j}\odot \boldsymbol{b}^{j}$ and from \eqref{equality} the consensus constraint $\boldsymbol{Y}^{i}=\boldsymbol{Y}^{j}$ and $\boldsymbol{P}^{i}=\boldsymbol{P}^{j}$. Combining the last two equalities yields $\boldsymbol{P}^{i}\odot (\boldsymbol{b}^{i}-\boldsymbol{b}^{j})=\boldsymbol{0}_{|\mathcal{A}|}$, which does not imply necessarily that $\boldsymbol{b}^{i}=\boldsymbol{b}^{j}$ when some components of the flow $\boldsymbol{P}^{i}$ are zeros. However, from an intuitive standpoint, if there is no flow passing through a particular $\left( i,j\right) \in \mathcal{A}$ and all the $P_{ij}^{k}=0$ for every $k\in V$, then the switch $\left( i,j\right) $ would be considered closed. Consequently, we have $b_{ij}^{k}=b_{ij}^{l}$ for all $k,l\in V$. This simplification is aimed at making the combinatorial sub-problem more manageable. As demonstrated in Section \ref{Sec:num}, this simplification has no impact on the algorithm's performance. 
	\end{remark}
	
	\section{The optimization algorithm}\label{Sec:algo}
	
	\subsection{Review of the MWRAP}
	
	Given a directed graph $G = (\mathcal{V}, \mathcal{A})$ and a vector of weights 
	$\left\lbrace h_{ij}\right\rbrace_{{(i,j)} \in \mathcal{A}}$, the MWRAP consists in finding an arborescence $(\mathcal{V},\mathcal{T})$, with $\mathcal{T} \subset \mathcal{A}$, rooted at a designated node in $G$, that minimizes the sum of weights $\sum_{(i,j)\in \mathcal{T}} h_{ij}$. For more details, we refer the reader to \cite[Chapter~6.2]{korte2018combinatorial}. By definition of the set $\mathbb{S}$, the MWRAP can also be formulated as an integer linear programming problem as the following
	\begin{equation}
		\underset{\boldsymbol{b}\in\mathbb{S}}{\arg\min}\;\boldsymbol{h}^T\boldsymbol{b},
		\label{MST}
	\end{equation}  
	Solving \eqref{MST} as a linear integer problem in every iteration reduces the effectiveness of ADMM as a heuristic approach to our problem.  This is attributed to the prolonged convergence time, particularly evident in large-scale networks. In practical scenarios, the MWRAP can be effectively solved by algorithms like Edmond's algorithm, with a runtime of $O(|\mathcal{A}|+|\mathcal{V}|\log(|\mathcal{V}|))$, as shown in \cite[Chapter~6.3]{korte2018combinatorial}.

	\subsection{The proposed algorithm}
	
	Let us introduce the auxiliary variable $\boldsymbol{k} \in \mathbb{R}^{4|\mathcal{A}|+|V|}$ such that
	\begin{equation*}
		\boldsymbol{X}^{i} = \boldsymbol{X}^{j} = \boldsymbol{k}^{ij}, \quad \forall \left( i, j \right) \in \mathcal{A}.
	\end{equation*}
	By considering these new constraints, the augmented Lagrangian associated with problem~\eqref{consensus} reads
	\begin{align*}
		&\mathcal{L}_{\delta} \left( 
		\begin{array}{c}
			\{\boldsymbol{X}^{i}\}, \{\boldsymbol{b}^{i}\}, \{\boldsymbol{k}^{ij}\}, \left\{ \boldsymbol{\alpha}^{i} \right\}, \\ 
			\left\{ \boldsymbol{\beta}^{i} \right\}, \{\boldsymbol{\gamma}^{i}\}, \{\boldsymbol{s}^{ij}\}, \{\boldsymbol{g}^{ij}\}%
		\end{array}
		\right) = \sum_{i \in \mathcal{V}} f_{i} (\boldsymbol{X}^{i}) \\
		&+ \frac{\delta}{2} \sum_{i \in \mathcal{V}} H_i (\boldsymbol{X}^{i}, \boldsymbol{b}^{i}, \delta^{-1} \boldsymbol{\alpha}^{i}, \delta^{-1} \boldsymbol{\beta}^{i}, \delta^{-1} \boldsymbol{\gamma}^{i})  \\
		&- 2 \delta^{-1} \sum_{i \in \mathcal{V}} \left\Vert \boldsymbol{\alpha}^{i} \right\Vert^{2}- 2 \delta^{-1} \sum_{i \in \mathcal{V}} \left\Vert \boldsymbol{\beta}^{i} \right\Vert^{2} - 2 \delta^{-1} \sum_{i \in \mathcal{V}} \left\Vert \boldsymbol{\gamma}^{i} \right\Vert^{2} \\
		&+ \sum_{(i,j) \in \mathcal{A}} \left( \boldsymbol{s}^{ij} \right)^{T} \left( \boldsymbol{X}^{i} - \boldsymbol{k}^{ij} \right) + \sum_{(i,j) \in \mathcal{A}} \left( \boldsymbol{g}^{ij} \right)^{T} \left( \boldsymbol{X}^{j} - \boldsymbol{k}^{ij} \right) \\
		&+ \frac{\delta}{2} \sum_{(i,j) \in \mathcal{A}} \left\Vert \boldsymbol{X}^{i} - \boldsymbol{k}^{ij} \right\Vert^{2} + \frac{\delta}{2} \sum_{(i,j) \in \mathcal{A}} \left\Vert \boldsymbol{X}^{j} - \boldsymbol{k}^{ij} \right\Vert^{2}
	\end{align*}
	where $\left\{ \boldsymbol{\alpha}^{i}\right\}_{i\in\mathcal{V}}$, $\left\{ \boldsymbol{\beta}^{i} \right\}_{i\in\mathcal{V}}$, $\left\{\boldsymbol{\gamma}^{i} \right\}_{i\in\mathcal{V}}$, $\{\boldsymbol{s}^{ij}\}_{(i,j) \in\mathcal{A}}$, $\{\boldsymbol{g}^{ij}\}_{(i,j) \in \mathcal{A}}$, are the Lagrange multipliers and $H_i$ is given by   
	\begin{align}
		&H_i(\boldsymbol{X}^i,\boldsymbol{b}^i,\boldsymbol{\alpha }^i,\boldsymbol{\beta}^i, \boldsymbol{\gamma }^i)   =\frac{1}{2}\left\Vert \boldsymbol{P}^i\odot \boldsymbol{b}^i-\boldsymbol{Y}^i+ 
		\boldsymbol{\alpha}^i\right\Vert ^{2}  \notag \\
		&+\frac{1}{2}\left\Vert \boldsymbol{Q}^i\odot \boldsymbol{b}^i-\boldsymbol{Z}^i+ 
		\boldsymbol{\beta}^i\right\Vert^{2}  \label{H}\\
		&+\frac{1}{2}\left\Vert \boldsymbol{b}^i \odot\boldsymbol{A}^{i}\boldsymbol{U}^i
		-2\left( \boldsymbol{r}^i\odot \boldsymbol{Y}^i+\boldsymbol{x}^i\odot 
		\boldsymbol{Z}^i\right) +\boldsymbol{\gamma}^i\right\Vert ^{2},  \notag
	\end{align}
	The ADMM iterates are given by
	\begin{eqnarray*}
		\boldsymbol{X}_{k+1}^{i} &\in &\underset{\boldsymbol{X}^{i} \in \mathbb{X}}{\arg \min} \ \mathcal{L}_{\delta} \left( 
		\begin{array}{c}
			\{\boldsymbol{X}^{i}\}, \{\boldsymbol{k}_{k}^{ij}\}, \{\boldsymbol{b}_{k}^{i}\}, \left\{ \boldsymbol{\alpha}_{k}^{i} \right\}, \\ 
			\left\{ \boldsymbol{\beta}_{k}^{i} \right\}, \{\boldsymbol{\gamma}_{k}^{i}\}, \{\boldsymbol{s}_{k}^{ij}\}, \{\boldsymbol{g}_{k}^{ij}\}%
		\end{array}
		\right),
	\end{eqnarray*}
	\begin{eqnarray*}
		\boldsymbol{k}_{k+1}^{ij} &\in &\underset{\boldsymbol{k}^{ij}}{\arg \min} \ \mathcal{L}_{\delta} \left( 
		\begin{array}{c}
			\{\boldsymbol{X}_{k+1}^{i}\}, \{\boldsymbol{k}^{ij}\}, \{\boldsymbol{b}_{k}^{i}\}, \left\{ \boldsymbol{\alpha}_{k}^{i} \right\}, \\ 
			\left\{ \boldsymbol{\beta}_{k}^{i} \right\}, \{\boldsymbol{\gamma}_{k}^{i}\}, \{\boldsymbol{s}_{k}^{ij}\}, \{\boldsymbol{g}_{k}^{ij}\}
		\end{array}
		\right),
	\end{eqnarray*}
	\begin{eqnarray*}
		\boldsymbol{b}_{k+1}^{i} &\in &\underset{\boldsymbol{b}^{i} \in \mathbb{S}}{\arg \min} \ \mathcal{L}_{\delta} \left( 
		\begin{array}{c}
			\{\boldsymbol{X}_{k+1}^{i}\}, \{\boldsymbol{k}_{k+1}^{ij}\}, \{\boldsymbol{b}^{i}\}, \left\{ \boldsymbol{\alpha}_{k}^{i} \right\}, \\ 
			\left\{ \boldsymbol{\beta}_{k}^{i} \right\}, \{\boldsymbol{\gamma}_{k}^{i}\}, \{\boldsymbol{s}_{k}^{ij}\}, \{\boldsymbol{g}_{k}^{ij}\}
		\end{array}
		\right),
	\end{eqnarray*}
	\begin{eqnarray*}
		\boldsymbol{\alpha}_{k+1}^{i} &=& \boldsymbol{\alpha}_{k}^{i} + \delta \left( \boldsymbol{Y}_{k+1}^{i} - \boldsymbol{P}_{k+1}^{i} \odot \boldsymbol{b}_{k+1}^{i} \right), \\
		\boldsymbol{\beta}_{k+1}^{i} &=& \boldsymbol{\beta}_{k}^{i} + \delta \left( \boldsymbol{Z}_{k+1}^{i} - \boldsymbol{Q}_{k+1}^{i} \odot \boldsymbol{b}_{k+1}^{i} \right), \\
		\boldsymbol{\gamma}_{k+1}^{i} &=& \boldsymbol{\gamma}_{k}^{i} + \delta \left( 
		\begin{array}{c}
			\boldsymbol{b}_{k+1}^{i} \odot \boldsymbol{A}^{i} \boldsymbol{U}_{k+1}^{i} \\ 
			-2 \left( \boldsymbol{r}^i \odot \boldsymbol{Y}_{k+1}^{i} + \boldsymbol{x}^i \odot \boldsymbol{Z}_{k+1}^{i} \right)
		\end{array}
		\right), \\
		\boldsymbol{s}_{k+1}^{ij} &=& \boldsymbol{s}_{k}^{ij} + \delta \left( \boldsymbol{X}_{k+1}^{i} - \boldsymbol{k}_{k+1}^{ij} \right), \\
		\boldsymbol{g}_{k+1}^{ij} &=& \boldsymbol{g}_{k}^{ij} + \delta \left( \boldsymbol{X}_{k+1}^{j} - \boldsymbol{k}_{k+1}^{ij} \right).
	\end{eqnarray*}
	A closed form can be obtained for $\boldsymbol{k}_{k+1}^{ij}$. Indeed, since the variable $\boldsymbol{k}^{ij}$ is unconstrained, we solve $\nabla_{\boldsymbol{k}_{k+1}^{ij}}\mathcal{L}_{\delta}=\boldsymbol{0}_{|\mathcal{A}|}$ to get
	\begin{equation*}
		\boldsymbol{0}_{|\mathcal{A}|} = - \left( \boldsymbol{s}_{k}^{ij} + \boldsymbol{g}_{k}^{ij} \right) + 2 \delta \boldsymbol{k}_{k+1}^{ij} - \delta \boldsymbol{X}_{k+1}^{i} - \delta \boldsymbol{X}_{k+1}^{j},
	\end{equation*}
	that is
	\begin{equation}
		\boldsymbol{k}_{k+1}^{ij} = \frac{1}{2 \delta} \left( \boldsymbol{s}_{k}^{ij} + \boldsymbol{g}_{k}^{ij} \right) + \frac{1}{2} \left( \boldsymbol{X}_{k+1}^{i} + \boldsymbol{X}_{k+1}^{j} \right). \label{z}
	\end{equation}
	Inserting \eqref{z} in the dual updates formulas we get
	\begin{eqnarray*}
		\boldsymbol{s}_{k+1}^{ij} &=& \frac{1}{2} \left( \boldsymbol{s}_{k}^{ij} - \boldsymbol{g}_{k}^{ij} \right) + \frac{\delta}{2} \left( \boldsymbol{X}_{k+1}^{i} - \boldsymbol{X}_{k+1}^{j} \right), \\
		\boldsymbol{g}_{k+1}^{ij} &=& \frac{1}{2} \left( \boldsymbol{g}_{k}^{ij} - \boldsymbol{s}_{k}^{ij} \right) + \frac{\delta}{2} \left( \boldsymbol{X}_{k+1}^{j} - \boldsymbol{X}_{k+1}^{i} \right).
	\end{eqnarray*}
	By taking the sum of the above formulas we obtain $\boldsymbol{s}_{k}^{ij} + \boldsymbol{g}_{k}^{ij} = \boldsymbol{0}_{|\mathcal{A}|}$, $k \geq 1$, which implies that
	\begin{equation}
		\boldsymbol{k}_{k+1}^{ij} = \frac{1}{2} \left( \boldsymbol{X}_{k+1}^{i} + \boldsymbol{X}_{k+1}^{j} \right). \label{w}
	\end{equation}
	The dual updates become
	\begin{eqnarray*}
		\boldsymbol{s}_{k+1}^{ij} &=& \boldsymbol{s}_{k}^{ij} + \frac{\delta}{2} \left( \boldsymbol{X}_{k+1}^{i} - \boldsymbol{X}_{k+1}^{j} \right), \\
		\boldsymbol{g}_{k+1}^{ij} &=& \boldsymbol{g}_{k}^{ij} + \frac{\delta}{2} \left( \boldsymbol{X}_{k+1}^{j} - \boldsymbol{X}_{k+1}^{i} \right).
	\end{eqnarray*}
	Defining the new multiplier $\left\{ \boldsymbol{\lambda}_{i} \right\}_{i \in \mathcal{V}}$ such that
	\begin{equation}
		\boldsymbol{\lambda}_{k+1}^{i} = \sum_{j:(i,j) \in \mathcal{A}} \boldsymbol{s}_{k+1}^{ij} + \sum_{j:(j,i) \in \mathcal{A}} \boldsymbol{g}_{k+1}^{ji}, \quad i \in \mathcal{V}, \label{lambda}
	\end{equation}
	then
	\begin{equation*}
		\boldsymbol{\lambda}_{k+1}^{i} = \boldsymbol{\lambda}_{k}^{i} + \delta \sum_{j \in N(i)} \left( \boldsymbol{X}_{k+1}^{i} - \boldsymbol{X}_{k+1}^{j} \right), \quad i \in \mathcal{V}.
	\end{equation*}
	So, the iterate with respect to the variable $\boldsymbol{X}_{k+1}^{i}$ writes
	\begin{eqnarray}
		\boldsymbol{X}_{k+1}^{i}& =& \underset{\boldsymbol{X}^{i} \in \mathbb{X}}{\arg \min} \sum_{i \in V} f_{i} (\boldsymbol{X}^{i}) \label{piz} \\
		&&+ \frac{\delta}{2} H_i (\boldsymbol{X}^{i}, \boldsymbol{b}_{k}^{i}, \delta^{-1} \boldsymbol{\alpha}_{k}^{i}, \delta^{-1} \boldsymbol{\beta}_{k}^{i}, \delta^{-1} \boldsymbol{\gamma}_{k}^{i}) \notag \\
		&&+ (\boldsymbol{X}^{i})^T \boldsymbol{\lambda}_k^i + \delta \sum_{j \in N(i)} \left\Vert \boldsymbol{X}^{i} - \frac{\boldsymbol{X}_{k}^{i} + \boldsymbol{X}_{k}^{j}}{2} \right\Vert^{2}. \notag
	\end{eqnarray}
	Finally, by inserting \eqref{lambda} in \eqref{piz} and rescaling the dual variables, we get the following ADMM iterates to solve problem~\eqref{consensus} for each agent $i\in V$:
	\begin{subequations}
		\begin{eqnarray}
			\boldsymbol{X}_{k+1}^{i} &\in &\underset{\boldsymbol{X}^{i}\in \mathbb{X} }{\arg \min }\;\delta ^{-1}f_{i}\left( \boldsymbol{X}^{i}\right)+H_i(\boldsymbol{X}^{i},\boldsymbol{b}^{i}_k,\boldsymbol{\alpha}_{k}^{i}, 
			\boldsymbol{\beta }_{k}^{i},\boldsymbol{\gamma }_{k}^{i})  \notag \\
			&&+\left( \boldsymbol{\lambda }_{k}^{i}\right) ^{T}\boldsymbol{X}
			^{i}+\sum_{j\in N(i)}\left\Vert \boldsymbol{X}^{i}-\frac{ \boldsymbol{X}_{k}^{i}+\boldsymbol{X}_{k}^{j}}{2}\right\Vert ^{2}  \label{admm1} \\
			\boldsymbol{b}_{k+1}^{i} &\in &\underset{\boldsymbol{b}^{i}\in \mathbb{S}}{
				\arg \min }\ H_i(\boldsymbol{X}_{k+1}^{i},\boldsymbol{b}^{i},\boldsymbol{%
				\alpha }_{k}^{i},\boldsymbol{\beta }_{k}^{i},\boldsymbol{\gamma }_{k}^{i}),
			\label{admm2} \\
			\boldsymbol{\alpha }_{k+1}^{i} &=&\boldsymbol{\alpha }_{k}^{i}+\boldsymbol{P}
			_{k+1}^{i}\odot \boldsymbol{b}_{k+1}^{i}-\boldsymbol{Y}_{k+1}^{i},
			\label{admm3} \\
			\boldsymbol{\beta }_{k+1}^{i} &=&\boldsymbol{\beta }_{k}^{i}+\boldsymbol{Q}
			_{k+1}^{i}\odot \boldsymbol{b}_{k+1}^{i}-\boldsymbol{Z}_{k+1}^{i},
			\label{admm4} \\
			\boldsymbol{\gamma }_{k+1}^{i} &=&\boldsymbol{\gamma}_{k}^{i}+\boldsymbol{b}_{k+1}^{i}\odot\boldsymbol{A}^{i}\boldsymbol{U}_{k+1}^{i} \label{admm5} \\
			&& -2\left( \boldsymbol{r}^i\odot \boldsymbol{Y}_{k+1}^{i}+\boldsymbol{x}^i\odot\boldsymbol{Z}_{k+1}^{i}\right) ,  \notag \\
			\boldsymbol{\lambda }_{k+1}^{i} &=&\boldsymbol{\lambda }_{k}^{i}+\sum_{j\in N(i)}\left( \boldsymbol{X}_{k+1}^{i}-\boldsymbol{X}_{k+1}^{j}\right) ,
			\label{admm6}
		\end{eqnarray}
	\end{subequations}
	where $N(i)$ is the neighbors set of the agent $i$, and $\delta >0$ is a quadratic penalty parameter. Iterates \eqref{admm1} and \eqref{admm2} are known as the primal updates while iterates \eqref{admm3}--\eqref{admm6} are known as the dual ones. Note that the iterates \eqref{admm1}--\eqref{admm6} are  distributed, namely, each agent $i$ exclusively exchanges information with its neighboring agents.
	
	It can be seen that problem~\eqref{admm1} is a linearly constrained convex quadratic problem that can be solved by using optimization solvers, while problem~\eqref{admm2} for $\boldsymbol{b}^{i}$, $i\in \mathcal{V}$, is a mixed-integer quadratic problem that is generally costly to solve. In the  Next, we show that it can be seen as an MWRAP with specific local weights.
	
	\begin{proposition}\label{Proposition} 
		For each $i\in \mathcal{V}$, solving problem~\eqref{admm2} is equivalent to solve the MWRAP on $G$ rooted at $g$ with weights 
		\begin{eqnarray}
			\boldsymbol{h}^{i} &=&\boldsymbol{h}_{1}^{i}+\boldsymbol{h}_{2}^{i}+%
			\boldsymbol{h}_{3}^{i}  \label{weights} \\
			&=&\boldsymbol{P}_{k+1}^{i}\odot \left( \boldsymbol{P}_{k+1}^{i}+2\left( 
			\boldsymbol{\alpha }_{k}^{i}-\boldsymbol{Y}_{k+1}^{i}\right) \right)   \notag
			\\
			&&+\boldsymbol{Q}_{k+1}^{i}\odot \left( \boldsymbol{Q}_{k+1}^{i}+2\left( 
			\boldsymbol{\beta }_{k}^{i}-\boldsymbol{Z}_{k+1}^{i}\right) \right)   \notag
			\\
			&&+\boldsymbol{A}^{i}\boldsymbol{U}_{k+1}^{i}\odot \left( 
			\begin{array}{c}
				\boldsymbol{A}^{i}\boldsymbol{U}_{k+1}^{i}+2\boldsymbol{\gamma }_{k}^{i} \\ 
				-4\left( \boldsymbol{r}^{i}\odot \boldsymbol{Y}_{k+1}^{i}+\boldsymbol{x}%
				^{i}\odot \boldsymbol{Z}_{k+1}^{i}\right) 
			\end{array}
			\right).   \notag
		\end{eqnarray}
	\end{proposition}
	
	\begin{proof}
		By expending the first term of $H_{i}(\boldsymbol{X}_{k+1}^{i},\boldsymbol{b}^{i},\boldsymbol{\alpha }_{k}^{i},\boldsymbol{\beta }_{k}^{i},\boldsymbol{\gamma }_{k}^{i})$ we get 
		\begin{eqnarray}
			&&\left\Vert \boldsymbol{P}_{k+1}^{i}\odot \boldsymbol{b}^{i}-\boldsymbol{Y}
			_{k+1}^{i}+\boldsymbol{\alpha}_{k}^{i}\right\Vert ^{2}  \label{r1} \\
			&=&\left\Vert \boldsymbol{P}_{k+1}^{i}\odot \boldsymbol{b}^{i}\right\Vert
			^{2}+2\left( \boldsymbol{P}_{k+1}^{i}\odot \boldsymbol{b}^{i}\right)
			^{T}\left( \boldsymbol{\alpha }_{k}^{i}-\boldsymbol{Y}_{k+1}^{i}\right)  
			\notag \\
			&&+\left\Vert \boldsymbol{\alpha }_{k}^{i}-\boldsymbol{Y}_{k+1}^{i}\right\Vert ^{2}  \notag \\
			&=&\left( \boldsymbol{P}_{k+1}^{i}\odot \boldsymbol{P}_{k+1}^{i}\right) ^{T}%
			\boldsymbol{b}^{i}+2\left( \boldsymbol{P}_{k+1}^{i}\odot \left( \boldsymbol{\alpha }_{k}^{i}-\boldsymbol{Y}_{k+1}^{i}\right) \right) ^{T}\boldsymbol{b}^{i}  \notag \\
			&&+\left\Vert \boldsymbol{\alpha }_{k}^{i}-\boldsymbol{Y}_{k+1}^{i}\right\Vert ^{2}  \notag \\
			&=&\left( \boldsymbol{h}_{1}^{i}\right) ^{T}\boldsymbol{b}^{i}+\left\Vert 
			\boldsymbol{\alpha }_{k}^{i}-\boldsymbol{Y}_{k+1}^{i}\right\Vert ^{2}. 
			\notag
		\end{eqnarray}
		Similarly, by expending the second term of $H_{i}(\boldsymbol{X}_{k+1}^{i},\boldsymbol{b}^{i},\boldsymbol{\alpha }_{k}^{i},\boldsymbol{\beta }_{k}^{i},\boldsymbol{\gamma }_{k}^{i})$ We obtain 
		\begin{equation}
			\left\Vert \boldsymbol{Q}_{k+1}^{i}\odot \boldsymbol{b}^{i}-\boldsymbol{Z}%
			_{k+1}^{i}+\boldsymbol{\beta }_{k}^{i}\right\Vert ^{2}=\left( \boldsymbol{h}%
			_{2}^{i}\right) ^{T}\boldsymbol{b}^{i}+\left\Vert \boldsymbol{\beta }%
			_{k}^{i}-\boldsymbol{Z}_{k+1}^{i}\right\Vert ^{2}  \label{r2}
		\end{equation}%
		For the third term, we have 
		\begin{eqnarray}
			&&\left\Vert 
			\begin{array}{c}
				\boldsymbol{b}^{i}\odot \boldsymbol{A}^{i}\boldsymbol{U}_{k+1}^{i}+%
				\boldsymbol{\gamma }_{k}^{i} \\ 
				-2\left( \boldsymbol{r}^{i}\odot \boldsymbol{Y}_{k+1}^{i}+\boldsymbol{x}%
				^{i}\odot \boldsymbol{Z}_{k+1}^{i}\right) 
			\end{array}%
			\right\Vert ^{2}  \label{r3} \\
			&=&\left\Vert \boldsymbol{b}^{i}\odot \boldsymbol{A}^{i}\boldsymbol{U}%
			_{k+1}^{i}\right\Vert ^{2}  \notag \\
			&&+\left\Vert \boldsymbol{\gamma }_{k}^{i}-2\left( \boldsymbol{r}^{i}\odot 
			\boldsymbol{Y}_{k+1}^{i}+\boldsymbol{x}^{i}\odot \boldsymbol{Z}%
			_{k+1}^{i}\right) \right\Vert ^{2}  \notag \\
			&&+2\left( \boldsymbol{b}^{i}\odot \boldsymbol{A}^{i}\boldsymbol{U}%
			_{k+1}^{i}\right) ^{T}\left( \boldsymbol{\gamma }_{k}^{i}-2\left( 
			\begin{array}{c}
				\boldsymbol{r}^{i}\odot \boldsymbol{Y}_{k+1}^{i} \\ 
				+\boldsymbol{x}^{i}\odot \boldsymbol{Z}_{k+1}^{i}%
			\end{array}%
			\right) \right)   \notag \\
			&=&\left( 
			\begin{array}{c}
				\boldsymbol{A}^{i}\boldsymbol{U}_{k+1}^{i}\odot \boldsymbol{A}^{i}%
				\boldsymbol{U}_{k+1}^{i}+ \\ 
				2\boldsymbol{A}^{i}\boldsymbol{U}_{k+1}^{i}\odot \left( \boldsymbol{\gamma }%
				_{k}^{i}-2\left( 
				\begin{array}{c}
					\boldsymbol{r}^{i}\odot \boldsymbol{Y}_{k+1}^{i} \\ 
					+\boldsymbol{x}^{i}\odot \boldsymbol{Z}_{k+1}^{i}%
				\end{array}%
				\right) \right) 
			\end{array}%
			\right) ^{T}\boldsymbol{b}^{i}  \notag \\
			&&+\left\Vert \boldsymbol{\gamma }_{k}^{i}-2\left( \boldsymbol{r}^{i}\odot 
			\boldsymbol{Z}_{k+1}^{i}+\boldsymbol{x}^{i}\odot \boldsymbol{Z}%
			_{k+1}^{i}\right) \right\Vert ^{2}  \notag \\
			&=&\left( \boldsymbol{h}_{3}^{i}\right) ^{T}\boldsymbol{b}^{i}+\left\Vert 
			\boldsymbol{\gamma }_{k}^{i}-2\left( \boldsymbol{r}^{i}\odot \boldsymbol{Z}%
			_{k+1}^{i}+\boldsymbol{x}^{i}\odot \boldsymbol{Z}_{k+1}^{i}\right)
			\right\Vert ^{2}  \notag
		\end{eqnarray}
		Above, we have used the fact that $\boldsymbol{b}^{i}\odot \boldsymbol{b}^{i}=\boldsymbol{b}^{i}$ for $\boldsymbol{b}^{i}\in \mathbb{S}$. By combining \eqref{r1}--\eqref{r3} we get (\ref{weights}). 
	\end{proof}
	
	So, at each iteration, each agent represented by $i\in \mathcal{V}$ is assigned the task of solving both the linearly constrained quadratic problem~\eqref{admm1} and a MWRAP with weights $\boldsymbol{h}^i$. The time complexity of the algorithm in each iteration is primarily determined by the quadratic problem since Edmond's algorithm has a lower time complexity in comparison.
	
	\begin{remark}
		We could solve problem~\eqref{admm1} as a linearly constrained quadratic programming problem incorporating constraint~\eqref{C3}, thereby eliminating the need for the update with respect to \(\boldsymbol{\gamma}_k^i\) in \eqref{admm5}. However, this approach might render problem~\eqref{admm1} infeasible for certain values of \(k\), causing the iteration to halt. By separating \eqref{admm5} from the constrained quadratic programming problem~\eqref{admm1}, we can avoid this issue.
	\end{remark}
	
	Let $\eta >0$ and $k_{\max}$ be an error tolerance parameter and the maximum number of iterations respectively. Let $e_{k}$, $k\geq 0$, be the successive change in the iterates \eqref{admm1}--\eqref{admm6}, defined by 
	\begin{equation}
		e_{k}=\sum_{i\in \mathcal{V}}\left( s_{k}^{i}+r_{k}^{i}+\sum_{j\in N(i)}\left\Vert 
		\boldsymbol{b}_{k}^{i}-\boldsymbol{b}_{k}^{j}\right\Vert \right) ,
		\label{error}
	\end{equation}
	where $s_{k}^{i}$ and $r_{k}^{i}$ are defined by  
	\begin{eqnarray}
		s_{k}^{i} &:=&\left\Vert \left( \boldsymbol{X},\boldsymbol{b}
		\right)_{k+1}^{i} -\left( \boldsymbol{X},\boldsymbol{b}\right)_{k}^{i}\right
		\Vert,  \notag \\
		r_{k}^{i} &:=&\left\Vert \left( \boldsymbol{\alpha },\boldsymbol{%
			\beta},\boldsymbol{\gamma},\boldsymbol{\lambda}\right)_{k+1}^{i}-\left(\boldsymbol{\alpha},\boldsymbol{\beta},\boldsymbol{\gamma},\boldsymbol{\lambda}\right)_{k}^{i}\right\Vert .  \notag
	\end{eqnarray}
	The term $\sum\limits_{j \in N(i)} \left\Vert \boldsymbol{b}_{k}^{i} - \boldsymbol{b}_{k}^{j} \right\Vert$ has been added \eqref{error} to guarantee the convergence of the local variables $\boldsymbol{b}^{i}$, $i\in \mathcal{V}$, to the same limit $\boldsymbol{b}$ (consensus) as there is no guarantee of its occurrence, see Remark~\ref{Remark}. Our algorithm is stated as follows:  
	\begin{algorithm}[H]
		\caption{Distributed Algorithm}
		\label{alg:Algorithm}
		\begin{algorithmic}[1]
			\REQUIRE $\boldsymbol{\rho}^1$, $\boldsymbol{\rho}^2$, $\boldsymbol{r}$, $\boldsymbol{x}$, $\delta>0$, $\eta>0$.
			
			\STATE Initialize $\boldsymbol{X}^i_0=\boldsymbol{\lambda}_0^i=\boldsymbol{0}_{4|\mathcal{A}|+|\mathcal{V}|}$, $\boldsymbol{\alpha}_0^i=\boldsymbol{\beta}^i_0=\boldsymbol{\gamma}^i_0=\boldsymbol{0}_{\mathcal{|A|}}$, $\boldsymbol{b}^i_0=\boldsymbol{b}_0\sim \mathcal{N}\left(0, \sigma \boldsymbol{I}_{|\mathcal{A}|}\right)$, $\forall i\in \mathcal{V}$, $k=0$.
			
			\REPEAT
			\FOR{each agent $i \in \mathcal{V}$}
			\STATE Update $\boldsymbol{X}_{k+1}^i$ according to \eqref{admm1}.
			\STATE Update $\boldsymbol{b}^i_{k+1}$ according to \eqref{admm2}.
			\STATE Update $\boldsymbol{\alpha}_{k+1}^i$, $\boldsymbol{\beta}_{k+1}^i$, $\boldsymbol{\gamma}_{k+1}^i$, $\boldsymbol{\lambda}_{k+1}^i$ according to \eqref{admm3}--\eqref{admm6}.
			\STATE $k \gets k+1$.
			\ENDFOR
			\UNTIL{$e_{k} < \eta$ or $k \geq k_{\max}$}
		\end{algorithmic}
	\end{algorithm}
	Since the convergence of \( e_k \) to zero in Algorithm~\ref{alg:Algorithm} is not guaranteed, we increase the likelihood of convergence by running this algorithm multiple times with different initial conditions \(\boldsymbol{b}_0\).
	
	\begin{remark}
		Examining the convergence of the proposed algorithm is a challenging task that falls outside the purview of this paper due to the mixed-integer nature of the problem. Instead, in the following section, we conduct various experiments to showcase both the effectiveness and limitations of our algorithms. It's crucial to emphasize that in the realm of non-convex optimization problems, ADMM does not assure of attaining an optimal, satisfactory, or even feasible solution within a finite number of iterations \cite{boyd-biconvex}. Additionally, it is noteworthy that ADMM's performance is contingent on the choice of initialization $\boldsymbol{b}^{0}$ and penalty parameter $\delta$, unlike the convex case where ADMM converges to an optimal solution regardless of the chosen $\boldsymbol{b}^{0}$ and $\delta$. 
	\end{remark}
	
	
	\section{Numerical Experiments}\label{Sec:num}
	\ym{	
		In this section, we evaluate the performance and limitations of Algorithm~\ref{alg:Algorithm} through a series of numerical simulations conducted across various test cases. The objective is to validate the algorithm’s robustness, adaptability, and scalability in both synthetic and realistic scenarios. The section is organized into three parts.}
	
	\ym{The first subsection investigates a small-scale synthetic network consisting of five nodes, designed to assess the algorithm’s resilience in response to dynamic changes such as topology modifications and variations in supply-demand profiles. We show that Algorithm~\ref{alg:Algorithm} successfully reconfigures the network and consistently converges to a radial solution under these perturbations.}
	
	\ym{The second subsection examines the 33-bus distribution test system, originally proposed in~\cite{baran1989network}. We analyze the convergence behavior of the algorithm under both nominal and faulted conditions. Additionally, a comparative study is presented, contrasting Algorithm~\ref{alg:Algorithm} with its centralized version proposed in \cite{mokhtari2025alternating} and the distributed algorithms developed in~\cite{shen2019distributed, nejad2021enhancing}. The comparison focuses on solution quality (optimality), convergence speed, and feasibility.}
	
	\ym{The third subsection demonstrates the scalability of Algorithm~\ref{alg:Algorithm} by applying it to a large-scale power distribution network operated by SRD, a French transmission system operator. This case study also includes a fault-handling experiment to assess the algorithm’s real-time resilience.}
	
	\ym{Unless otherwise specified, the tolerance parameter $\varepsilon$ for the error metric $e_k$, defined in~\eqref{error}, is fixed at $10^{-4} \times |\mathcal{V}|$, and the penalization parameter $\delta$ is set to $1$. The quadratic subproblem~\eqref{admm1} is solved using MATLAB's \textsf{quadprog} function. The MWRAP is solved using an implementation of Edmonds' algorithm, adapted from \cite{Choudhary:algo23:Edmond}.}
	
	
	\subsection{Resilience: a toy network}
	\subsubsection{Objective}\ym{The aim of this subsection is to demonstrate that Algorithm~\ref{alg:Algorithm} possesses the capability to adapt to dynamic changes within the network environment as they occur in real-time. These adaptations include the detection of modifications such as integrating a new source of energy, switches going offline, or shifts in network traffic patterns. The algorithm subsequently reconfigures the network in response to these changes during its execution without restarting.
		\subsubsection{Experience} This process is illustrated using a simple example involving a graph of size $|\mathcal{V}| = 5$ with one substation indexed by node $1$. The voltage constraint~\eqref{voltage} and its corresponding dual updates are omitted. We set the cost coefficients as $c_{ij} = i + j$ for all $(i,j) \in \mathcal{A}$. The experience is divided into four phases:}
	
	\begin{itemize}
		\item \ym{\textbf{Phase 1: $k=1$ to $k=100$:} Algorithm~\ref{alg:Algorithm} is executed with initial supply-demand vectors $\boldsymbol{\rho}^i = (4, -1, -1, -1, -1)^T$, $i=1,2$.}
		
		\item \ym{\textbf{Phase 2: $k=100$ to $k=200$:} A line failure is simulated on edge $(2,1)$ by assigning zero values to $b_{2,1}^{1}$, $b_{1,2}^{1}$, $b_{1,2}^{2}$, and $b_{2,1}^{2}$, with this change locally visible to agents 1 and 2.}
		
		\item \ym{\textbf{Phase 3: $k=200$ to $k=300$:} The supply-demand vector is modified to $\boldsymbol{\rho} = (4, -1, -2, 2, -3)^T$, increasing demand on buses 3 and 5 while incorporating renewable generation from node 4.}
		
		\item \ym{\textbf{Phase 4: $k=300$ to convergence:} The previously failed line $(1,2)$ is restored. The algorithm is then adjusted to accommodate this structural change while enforcing radiality constraints.}
	\end{itemize}
	\subsubsection{Results}
	\ym{ After running Algorithm~\ref{alg:Algorithm}, we obtained the following results:}
	
	\begin{itemize}
		\item \ym{In \textbf{Phase 1}, the algorithm converges efficiently to the optimal radial configuration under the initial load distribution, as depicted in Figure~\ref{Phase 1}.}
		
		\item \ym{In \textbf{Phase 2}, the algorithm detects the failure of line $(2,1)$ without global knowledge and successfully reconfigures the network. It converges to a new feasible radial structure as shown in Figure~\ref{Phase 2}.}
		
		\item \ym{In \textbf{Phase 3}, despite dynamic changes in demand and renewable supply, Algorithm~\ref{alg:Algorithm} remains robust and converges to the adjusted optimum, illustrated in Figure~\ref{Phase 3}.}
		
		\item \ym{In \textbf{Phase 4}, after restoring line $(1,2)$, the algorithm adapts and converges to a new feasible solution that incorporates the restored topology, while preserving radiality. The final configuration is illustrated in Figure~\ref{Phase 4}.}
	\end{itemize}
	\ym{The evolution of the iterative process, along with several key quantities across different phases, is illustrated in Figures~\ref{Fig:real-time1}--\ref{Fig:real-time4}. These figures demonstrate that all five agents progressively converge to a common solution, thereby confirming the achievement of consensus.}
	
	\subsubsection{Conclusion}
	\ym{This example highlights the algorithm’s resilience and adaptability. It efficiently handles real-time changes such as line faults and load shifts while maintaining feasibility, consensus, and radiality throughout all phases.}

	\begin{figure}[t!]
		\centering
		\par
		\begin{minipage}[b]{0.45\linewidth}
			\centering
			\subfloat[Phase 1]{\includegraphics[width=1\linewidth]{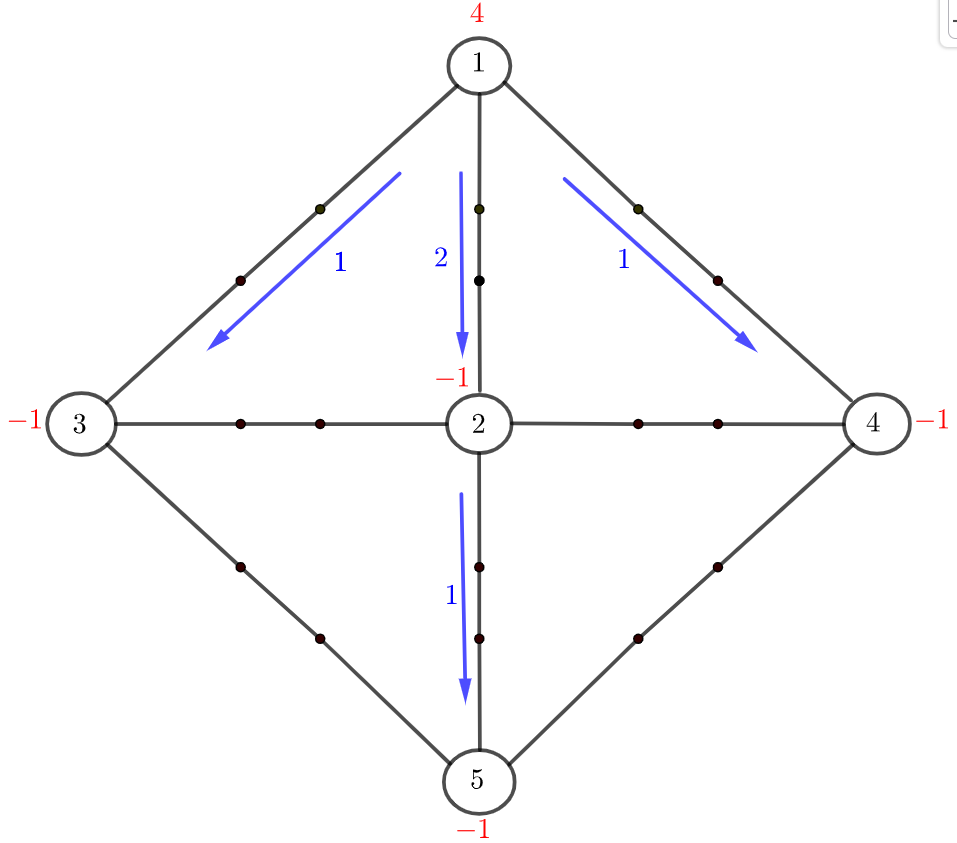}
				\label{Phase 1}}
		\end{minipage}
		\hfil
		\begin{minipage}[b]{0.45\linewidth}
			\centering
			\subfloat[Phase 2]{\includegraphics[width=1\linewidth]{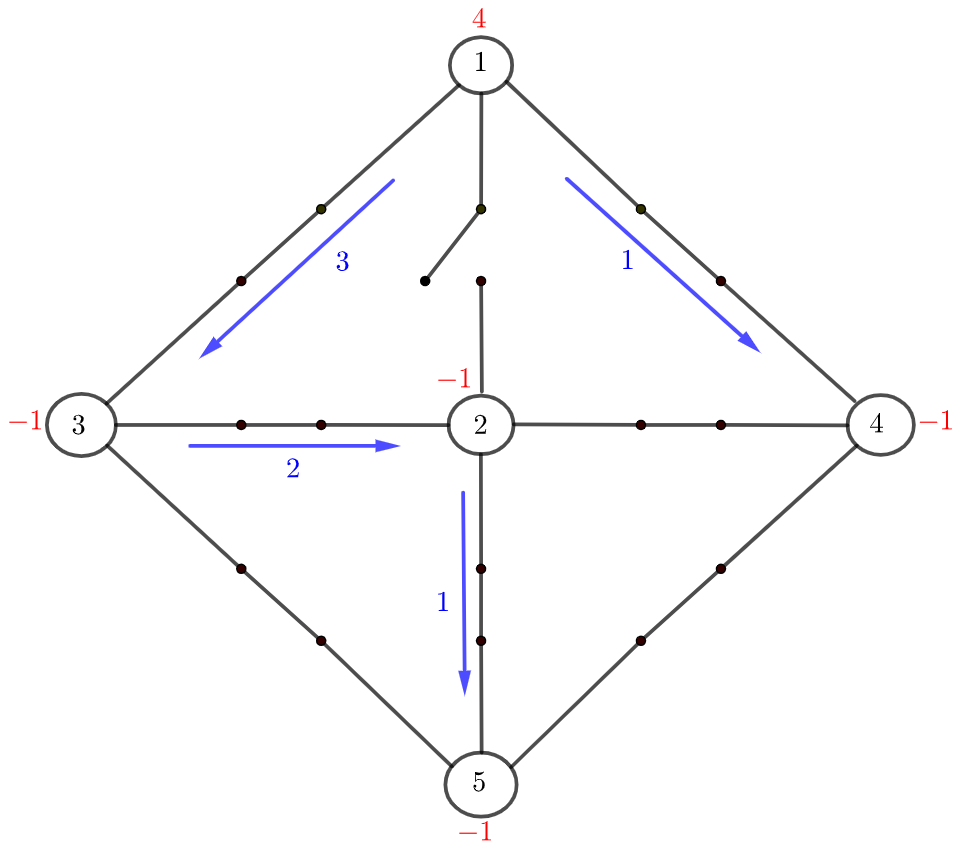}
				\label{Phase 2}}
		\end{minipage}
		\par
		\begin{minipage}[b]{0.45\linewidth}
			\centering
			\subfloat[Phase 3]{\includegraphics[width=1\linewidth]{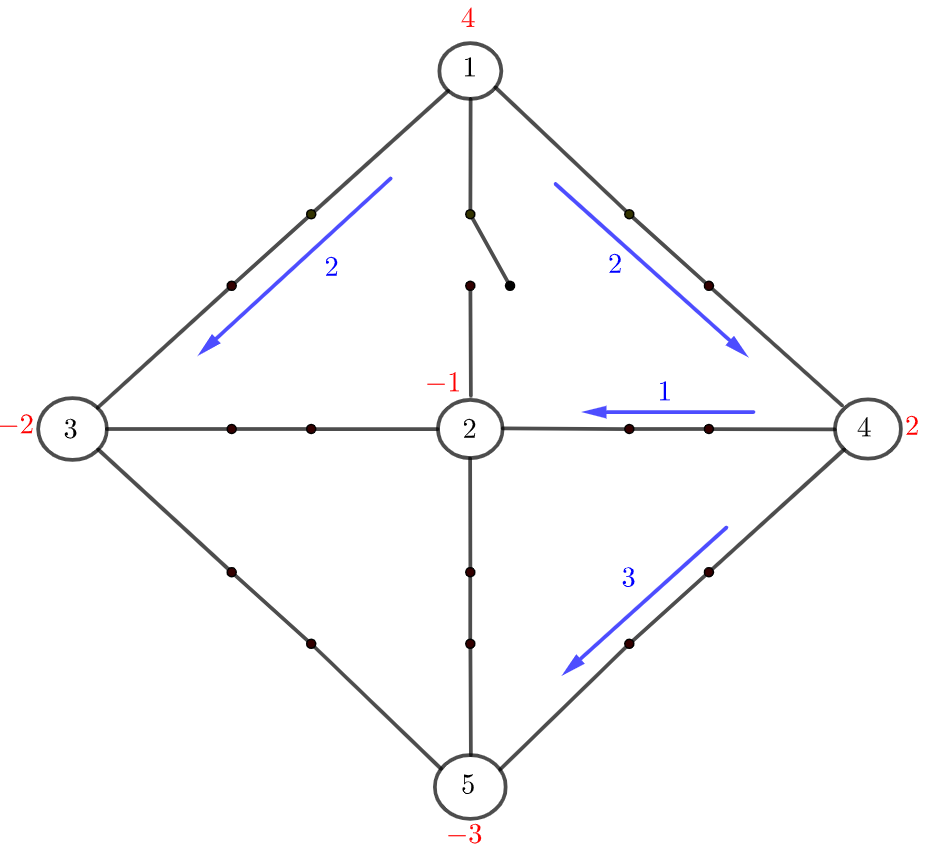}%
				\label{Phase 3}}
		\end{minipage}
		\hfil
		\begin{minipage}[b]{0.45\linewidth}
			\centering
			\subfloat[Phase 4]{\includegraphics[width=1\linewidth]{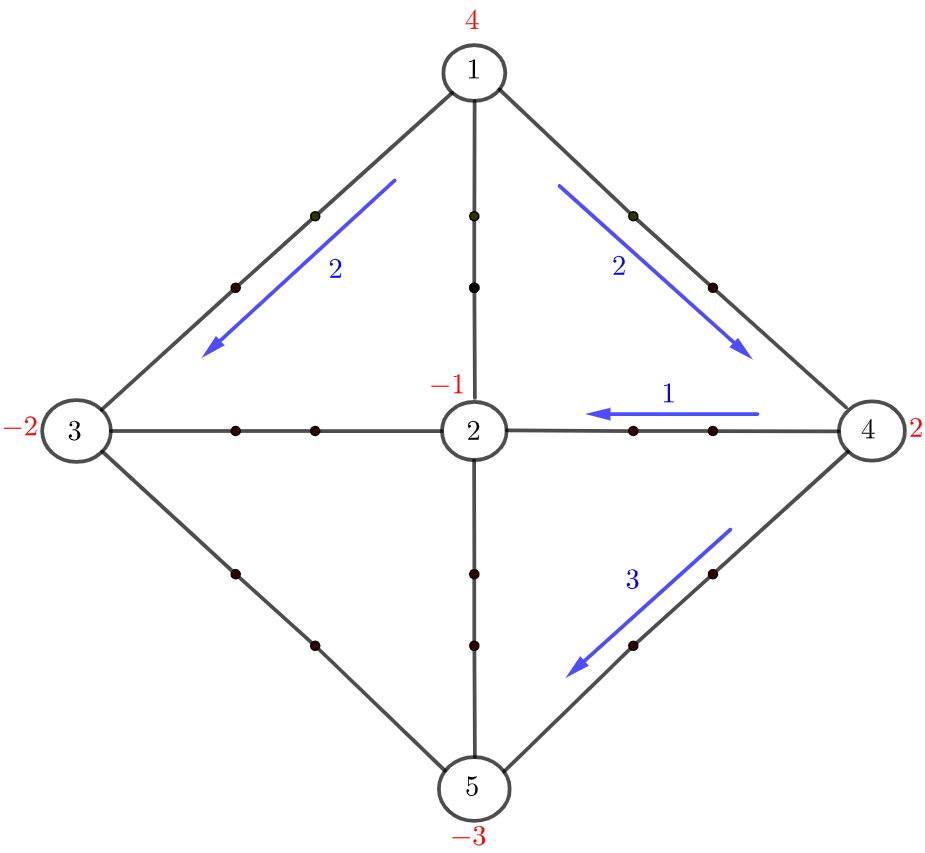}
				\label{Phase 4}}
		\end{minipage}
		\caption{The various phases of the experience. The supply-demand of nodes and the flow direction are highlighted in red and blue, respectively.}
		\label{fig_phases}
	\end{figure}
	
	\begin{figure}[t!]
		\centering
		\par
		\begin{minipage}[b]{0.45\linewidth}
			\centering
			\subfloat[Evolution of $e_k$ in function of iteration number.]{\includegraphics[width=\textwidth]{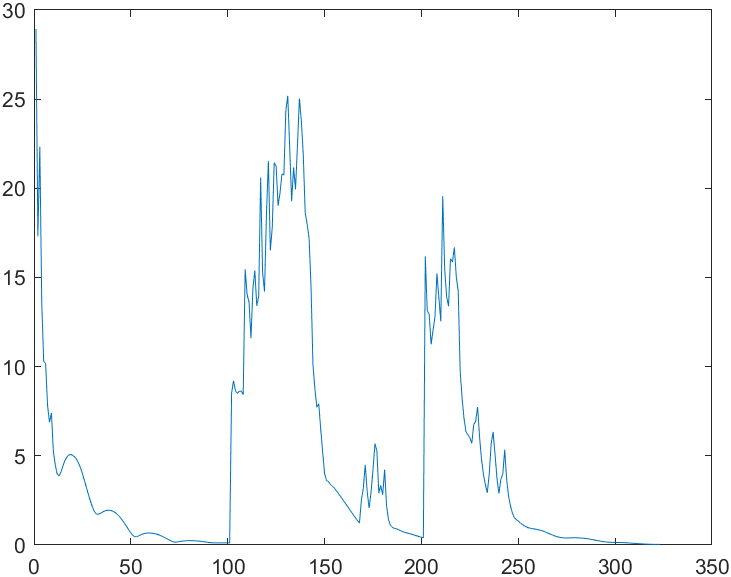}%
				\label{Fig:real-time1}}
		\end{minipage}
		\hfil
		\begin{minipage}[b]{0.45\linewidth}
			\centering
			\subfloat[Evolution of $\left\| \boldsymbol{P}_k^i \right\|^2$ in function of iteration number.]{\includegraphics[width=\textwidth]{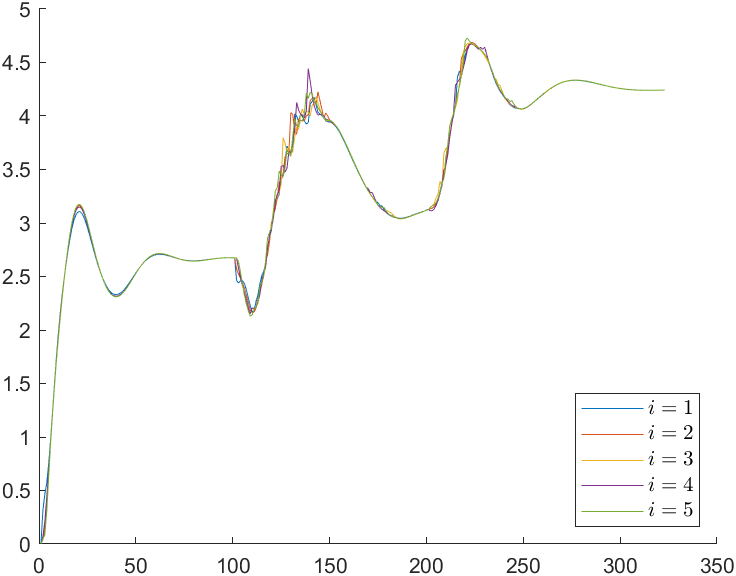}				\label{Fig:real-time4}}
		\end{minipage}
		\caption{The various plots show the evolution of different parameters in function of iteration number.}
		\label{fig_realtime}
	\end{figure}
	
	\subsection{Baran and Wu test case}
	\ym{\subsubsection{Objective}
		The objective of this subsection is to evaluate the performance of Algorithm~\ref{alg:Algorithm} on the 33-bus test system under four distinct scenarios:}
	
	\begin{itemize}
		\item \ym{\textbf{Scenario 1: Nominal Conditions.} This scenario evaluates the algorithm's ability to reach convergence and consensus when executed under normal conditions, with no structural changes or perturbations throughout the process.}
		
		\item \ym{\textbf{Scenario 2: Fault Resilience.} A fault is introduced by disabling a specific line in the network. This scenario tests the algorithm’s capacity to detect and adapt to the fault while converging to a new feasible solution that satisfies the network constraints.}
		
		\item \ym{\textbf{Scenario 3: Comparison with Centralized Algorithm.} The distributed version is compared with the centralized algorithm proposed in~\cite{mokhtari2025alternating}, with a focus on solution quality, iteration count, and execution time.}
		
		\item \ym{\textbf{Scenario 4: Comparison with Other Distributed Methods.} The proposed algorithm is benchmarked against the methods developed in~\cite{shen2019distributed} and~\cite{nejad2021enhancing}, as well as their enhanced versions, which will be introduced in the following subsection.}
	\end{itemize}
	
	Algorithm~1 has been tested on the Baran and Wu test case \cite	{baran1989network} which comprises $33$ nodes, $32$ branches, and $5$ tie lines, with a total active and reactive demand equal to $3.715$ MW and $2.24$ MVAr, respectively, and voltage base $12.66$ kV, and a single substation indexed by $1$ linked to a generator, as illustrated in Figure \ref{fig:33bus}. The data of the network can be found in \cite{baran1989network}. The optimization process involves comparing the outcome of Algorithm~\ref{alg:Algorithm} with an initial configuration defined by the inactivation of the following lines: $(25,29)$, $(18,33)$, $(9,15)$, $(8,21)$, $(12,22)$ resulting in a total approximate loss of $183.74$ kW (against $202.7$ kW by using the non-simplified model \cite{khodr2009distribution}). \ym{The result for each scenario is presented below.}

	\begin{figure}[t!]
		\centering
		\includegraphics[width=0.6\linewidth, height=0.2\textheight]{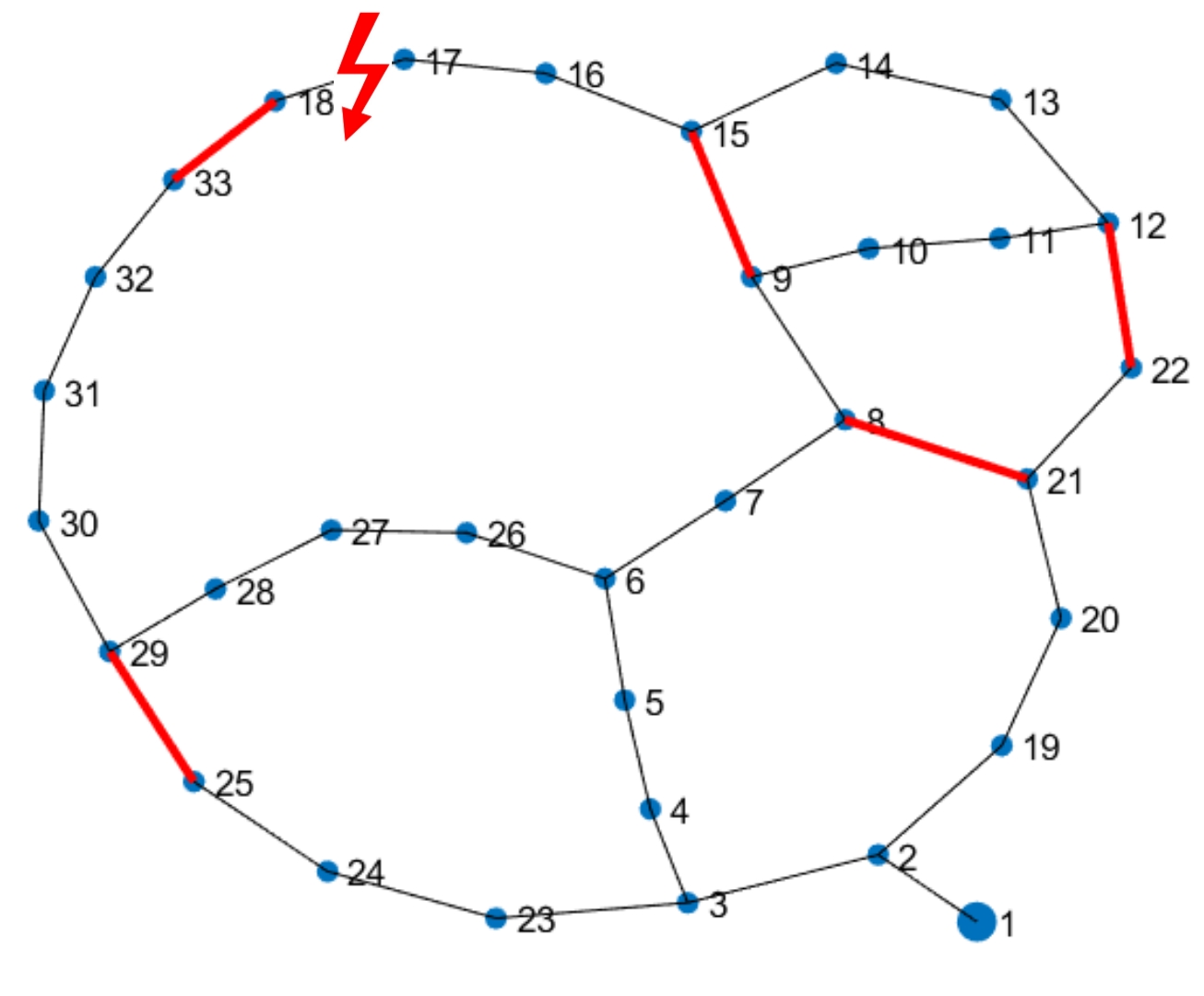} 
		\caption{Baran and Wu network. The lines in red are inactive in the initial
			configuration. The Fault occurs in the line $(17,18)$ after the $1200$th
			iteration.}
		\label{fig:33bus}
	\end{figure}
	\subsubsection{Convergence}
	
	As Algorithm~\ref{alg:Algorithm} is highly influenced by its initial conditions, it underwent ten executions, each with a distinct $\boldsymbol{b}_{0}$. Among these runs, it effectively determined the best solution, activating switches $(8,21)$, $(9,10)$, $(14,15)$, $(28,29)$, and $(32,33)$. This resulted in a loss of $132.28$ kW (against $139.5$ kW by using the non-simplified model \cite{khodr2009distribution}), with $\delta $ set at $0.1$. Achieving this outcome required $1349$ iterations, see Figure~\ref{fig:ek_nofault}. \ym{The consensus behavior is also illustrated in Figure~\ref{fig:xnorm_nofault}.}
	
	Notably, this optimal configuration was attained in only $2$ of the $10$ restarts. In the remaining runs, Algorithm~\ref{alg:Algorithm} converged but without reaching the optimal solution. See Table~\ref{Table1} for more details. We chose the $\delta $ value in this specific manner to increase the system's sensitivity to cost considerations over the quadratic penalization. It is essential to recognize, however, that increasing the value of $\delta $ can potentially sacrifice optimality while simultaneously reducing the number of iterations. 
	
	\subsubsection{Resilience}
	
	We demonstrate the resilience of Algorithm~\ref{alg:Algorithm} by introducing a fault in the line $(17,18)$ (any other line can be chosen) after the $1200$th iteration. This is done by setting $b_{(17,18)}^{17}=b_{(18,17)}^{17}=b_{(17,18)}^{18}=b_{(18,17)}^{18}=0$, ensuring that the fault is only known locally to agents $17$ and $18$. Under these conditions, Algorithm~\ref{alg:Algorithm} required $2889$ iterations to converge to the solution with open switches at $(7,8)$, $(10,11)$, $(14,15)$, $(17,18)$, and $(28,29)$, resulting in a loss of $157.8$ kW. The progression of the error $e_{k}$ \ym{and the agent's norm behavior are illustrated in Figures \ref{fig:ek_fault} and \ref{fig:xnorm_fault}}.
	
	\begin{table*}[ht!]
		\centering
		\caption{Execution results of Algorithm~\ref{alg:Algorithm} on the 33-bus system for various $\delta$ values. Averages are computed over 10 runs.}	\begin{tabular}{|l|c|c|c|c|}
			\hline
			& Average Losses (kW) & Min-Max Losses (kW) & Average No. of Iterations & Feasibility Ratio \\
			\hline
			$\delta=0.1$ & 140.1 & 132.28-143.11 & 1430 & $10/10$ \\
			$\delta=1$ & 156.43 & 143.76-163.18 & 1190 & $10/10$ \\
			$\delta=10$ & 172.3 & 162.9-180.2 & 871 & $10/10$ \\
			$\delta=100$ & 188.4 & 176.08-198.55 & 517 &  $10/10$ \\
			$\delta=10^8$ & 201.76 & 190.65-212.62 & 497 & $10/10$ \\
			\hline
		\end{tabular}
		\label{Table1}
	\end{table*}\begin{figure}[t!]
		\centering
		
		\subfloat[Error evolution $e_k$ without fault.]{%
			\includegraphics[width=0.45\linewidth]{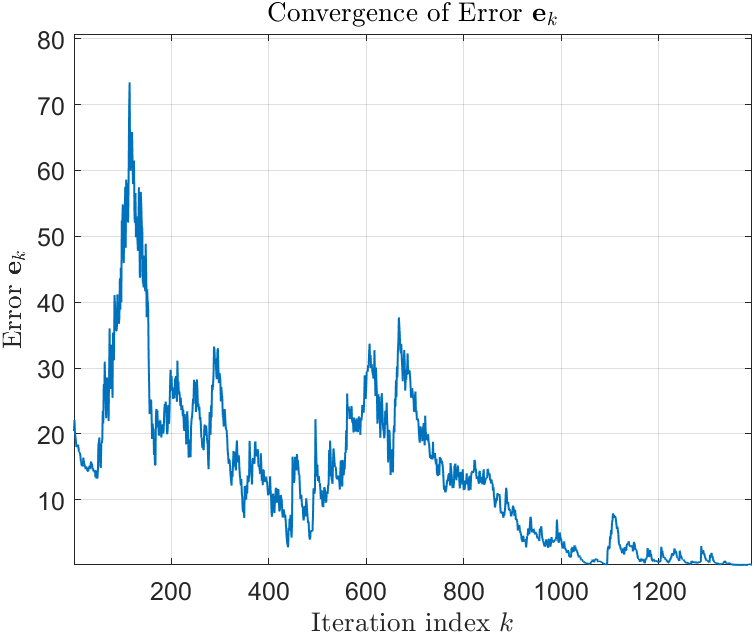}%
			\label{fig:ek_nofault}
		} \hfill
		\subfloat[Zoom-in on the $\ell^2$-norm of the 33 agents, $\|\boldsymbol{X}_k^{i}\|$ without fault.]{%
			\includegraphics[width=0.45\linewidth]{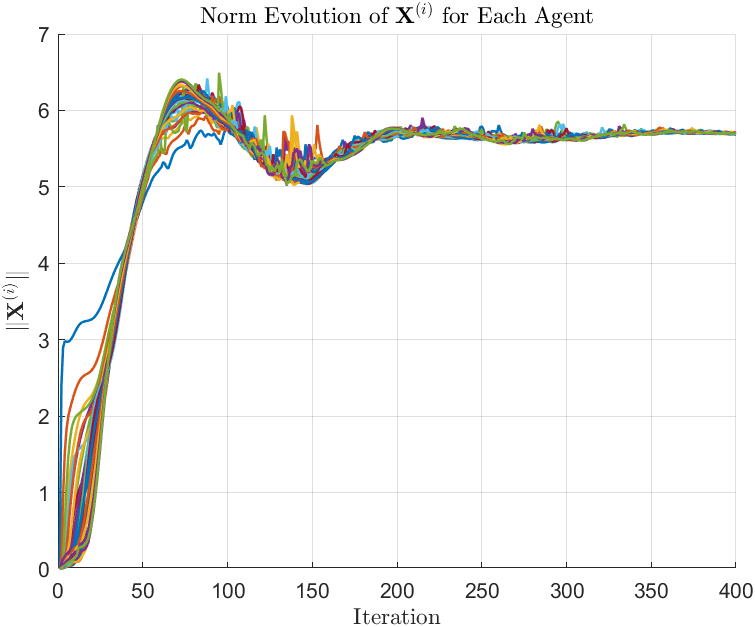}%
			\label{fig:xnorm_nofault}
		}
		
		\vspace{0.5em} 
		
		\subfloat[Error evolution $e_k$ with fault at iteration 1200.]{%
			\includegraphics[width=0.45\linewidth]{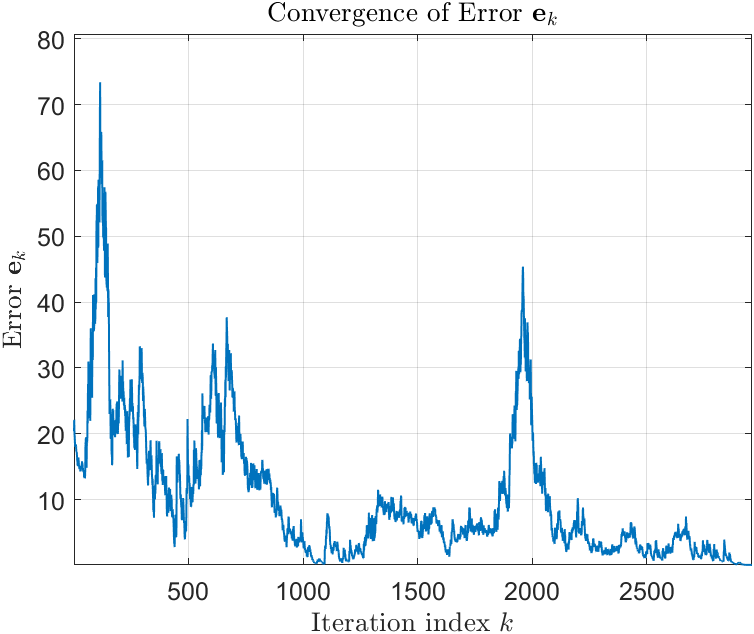}%
			\label{fig:ek_fault}
		} \hfill
		\subfloat[$\ell^2$-norm of the 33 agents $\|\boldsymbol{X}_k^{i}\|$ with fault at iteration 1200.]{%
			\includegraphics[width=0.45\linewidth]{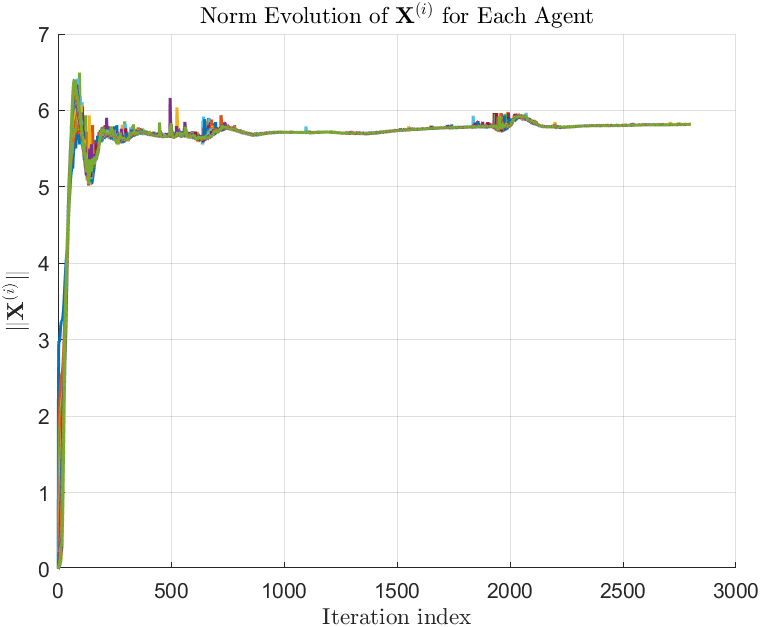}%
			\label{fig:xnorm_fault}
		}
		
		\caption{Evolution of the error $e_k$ and the agent norm $\|\boldsymbol{X}_k^{(i)}\|$ in scenarios without fault (top) and with a fault introduced at iteration 1200 (bottom). All the agents converge toward the same limit (consensus).}
		\label{fig:all_errors}
	\end{figure}

	\subsection{Centralized vs distributed}
	
	In this section, we compare Algorithm \ref{alg:Algorithm} with its centralized counterpart \ym{proposed in \cite{mokhtari2025alternating}}, which involves the following iterations:
	\begin{eqnarray*}
		\boldsymbol{X}^{k+1} &\in &\underset{\boldsymbol{X}\in \mathbb{X}}{\arg \min 
		}\ \delta ^{-1}\boldsymbol{r}^{T}\left( \boldsymbol{Y}\odot \boldsymbol{Y+Z} \odot \boldsymbol{Z}\right)  \\
		&&+H(\boldsymbol{X},\boldsymbol{b}_{k},\boldsymbol{\alpha }_{k},\boldsymbol{\beta }_{k},\boldsymbol{\lambda }_{k}\mathbb{)}, \\
		\boldsymbol{b}^{k+1} &\in &\underset{\boldsymbol{b\in }\mathbb{S}}{\arg \min 
		}\text{ }H(\boldsymbol{X}_{k+1},\boldsymbol{b},\boldsymbol{\alpha }_{k},	\boldsymbol{\beta }_{k},\boldsymbol{\lambda }_{k}\mathbb{)}, \\
		\boldsymbol{\alpha }^{k+1} &=&\boldsymbol{\alpha }^{k}+\boldsymbol{P}%
		^{k+1}\odot \boldsymbol{b}^{k+1}-\boldsymbol{Y}^{k+1}, \\
		\boldsymbol{\beta }^{k+1} &=&\boldsymbol{\beta }^{k}+\boldsymbol{Q}%
		^{k+1}\odot \boldsymbol{b}^{k+1}-\boldsymbol{Z}^{k+1}, \\
		\boldsymbol{\lambda }^{k+1} &=&\boldsymbol{\lambda }^{k}+\boldsymbol{b\odot
			AU}_{k+1} \\
		&&-2\left( \boldsymbol{r}\odot \boldsymbol{Y}_{k+1}+\boldsymbol{x}\odot 
		\boldsymbol{Z}_{k+1}\right) 
	\end{eqnarray*}%
	In this context, all vectors are defined similarly but without the agent's indices $i \in \mathcal{V}$. Similar to Proposition \ref{Proposition}, we can show that the problem for $\boldsymbol{b}$ is an MWRSA problem with the same weights, excluding the agents' indices.
	\begin{table*}[ht!]
		\centering
		\caption{Results for the centralized algorithm on the 33-bus system with different values of $\delta$.}
		\begin{tabular}{|l|c|c|c|c|}
			\hline
			& Average Losses (kW) & Min-Max Losses (kW) & Average No. of Iterations & Feasibility Ratio \\
			\hline
			$\delta=1$ & 137.8 & 132.09--140.87 & 146 & $10/10$ \\
			$\delta=10$ & 152.3 & 146.7--158.4 & 31 & $10/10$ \\
			$\delta=100$ & 166.8 & 152.12--188.7 & 14 & $10/10$ \\
			$\delta=10^8$ & 198.2 & 175.2--224 & 4 & $10/10$ \\
			\hline
		\end{tabular}
		\label{Table_centralized}
	\end{table*}
	
	For the centralized algorithm, we fix the error tolerance parameter at $10^{-4}$. 
	
	Table~\ref{Table_centralized} presents the results of the centralized algorithm on the 33-bus system. As indicated in Table~\ref{Table1}, the distributed algorithm requires significantly more iterations to reach convergence compared to the centralized version. This is due to the distributed approach's limitation, where agents only have information about their neighboring nodes.
	
	The centralized algorithm consistently achieves lower average losses compared to the distributed algorithm and exhibits more consistent performance with narrower min-max loss ranges. It requires fewer iterations and has shorter overall computation times compared to the global convergence time of the distributed algorithm. Specifically, the centralized algorithm takes 8.27 seconds on average, whereas the distributed algorithm takes 61.3 seconds for all 33 agents, or approximately 1.86 seconds per agent. Both algorithms maintain a perfect feasibility ratio. While the centralized algorithm provides more optimal and efficient solutions, the distributed algorithm has a shorter computation time per agent.

	\subsubsection{Algorithm \protect\ref{alg:Algorithm} Compared to \protect\cite{shen2019distributed} and \protect\cite{nejad2021enhancing}}
	
	In this section, we compare Algorithm~\ref{alg:Algorithm} with those proposed in \cite{shen2019distributed} and \cite{nejad2021enhancing}, where a relaxation technique is used instead of the variable substitution \eqref{substitution}. The combinatorial problem with respect to $\boldsymbol{b}^{i}$ for $i \in \mathcal{V}$ is defined as:
	\begin{equation}
		\boldsymbol{b}^{i} = \underset{\boldsymbol{b}^{i} \in \mathbb{S}}{\arg \min} \left\Vert \boldsymbol{w}_{k+1}^{i} - \boldsymbol{b}^{i} + \boldsymbol{u}_{k}^{i} \right\Vert ,  \label{P}
	\end{equation}
	where $\boldsymbol{w}^{i}\in [0,1]^{|\mathcal{A}|}$ is the relaxation of $\boldsymbol{b}^{i}$ and $\boldsymbol{u}^{i}$ is the Lagrange multiplier for the constraints $\boldsymbol{w}^{i} = \boldsymbol{b}^{i}$, $i \in \mathcal{V}$. The improvement of \cite{nejad2021enhancing} over \cite{shen2019distributed} lies in the level of the projection operator where the Douglas-Rachford splitting method was used. However, in both papers, problem~\eqref{P} is treated as a projection problem of $\boldsymbol{w}^{i}_{k+1} + \boldsymbol{u}_{k}^{i}$  onto the set $\{0,1\}^{|\mathcal{A}|}$, which does not ensure the tree structure of the solution. Instead, in the same spirit of Proposition~\ref{Proposition}, it can be shown that problem~\eqref{P} can be solved by solving an MWRAP with weights \ym{(see \cite[Proposition 3.2]{mokhtari2025alternating})}:
	\begin{equation}
		\boldsymbol{h}^{i} = -\boldsymbol{w}_{k+1}^{i} - \boldsymbol{u}_{k}^{i}, \quad i \in \mathcal{V}. \label{W}
	\end{equation}
	We refer to the algorithm from \cite{shen2019distributed} and \cite{nejad2021enhancing}, which handles problem~\eqref{P} by solving the MWRAP with weights, as "Algorithm \emph{relax}". Unfortunately, the algorithm from \cite{lopez2023enhanced} cannot be included in the comparison process since no explicit formula for the projection operator was provided. As before, we run Algorithm~\ref{alg:Algorithm}, Algorithm \emph{relax}, and the algorithms from \cite{shen2019distributed} and \cite{nejad2021enhancing} 10 times with different $\boldsymbol{b}_0$ then we take the average of these runs.
	
	\begin{table}[ht!]
		\centering
		\caption{Iteration counts and feasibility ratio over 10 runs of Algorithm~\ref{alg:Algorithm}, Algorithm \emph{relax}, and the algorithms from \cite{shen2019distributed} and \cite{nejad2021enhancing}. \\ 
			\hspace{\textwidth}¤ : heuristic did not converge after $k_{\max}=5000$. \\ 
			\hspace{\textwidth}* : the solution is not radial.}
		
		\begin{tabular}{|l|c|c|c|c|}
			\hline
			& Algorithm~\ref{alg:Algorithm}& Algorithm \emph{relax} &\cite{shen2019distributed}& \cite{nejad2021enhancing} \\
			\hline
			$\delta=0.1$ & 1430 – 10/10  & ¤ & ¤ &  ¤\\
			$\delta=1$ & 1190 – 10/10 & 4761 – 6/10 & ¤ & ¤ \\
			$\delta=10$ & 871 – 10/10 & 3981 – 8/10 &  ¤& * \\
			$\delta=100$ & 517 – 10/10 & 3560 – 8/10&  *& * \\
			$\delta=10^8$ & 497 – 10/10 & 2912 – 9/10 & * &*  \\
			\hline
		\end{tabular}
		\label{Table2}
		
	\end{table}
	Table~\ref{Table2} summarizes the performance of the various algorithms across different values of $\delta$, with each algorithm executed ten times. Several key observations can be drawn from this comparison:
	
	Algorithm~\ref{alg:Algorithm} consistently converges in all runs across different $\delta$ values, demonstrating robust performance with iteration counts decreasing as $\delta$ increases. In contrast, Algorithm \emph{relax} shows mixed results. It fails to converge for $\delta = 0.1$, $1$, and $10$. For higher values of $\delta$, it converges in fewer iterations but less consistently than Algorithm~\ref{alg:Algorithm}. The algorithms from \cite{shen2019distributed} and \cite{nejad2021enhancing} mostly fail to converge, especially for lower values of $\delta$. When they do converge, the solutions are not radial.
	
	For higher values of $\delta$ ($\delta = 10^8$), both Algorithm \emph{relax} and the referenced algorithms show improved feasibility ratios, but still not as high as Algorithm~\ref{alg:Algorithm}. This comparison indicates that Algorithm~\ref{alg:Algorithm} is superior in terms of consistency and iteration efficiency across different values of $\delta$. The use of a relaxation technique in Algorithm \emph{relax} and the methods from \cite{shen2019distributed} and \cite{nejad2021enhancing} appears less reliable than the variable substitution \eqref{substitution}, particularly for lower $\delta$ values.

	\subsubsection{Conclusion}
	\ym{The numerical results demonstrate that Algorithm~\ref{alg:Algorithm} effectively achieves consensus and optimal configurations in both standard and fault scenarios. While convergence to the global optimum depends on initialization, the algorithm consistently produces feasible and near-optimal solutions. Its resilience to faults and adaptability to changing conditions confirm its robustness.}
	
	\ym{Although the distributed algorithm requires more iterations than its centralized counterpart, it offers superior scalability and supports fully decentralized execution. Furthermore, it consistently outperforms relaxation-based distributed methods~\cite{shen2019distributed,nejad2021enhancing} in both feasibility and convergence reliability, owing to its exact projection step—formulated as a MWRAP which guarantees the radial solutions.}

	\subsection{Real Case Study: SRD Network}
	\subsubsection{Objective}
	\ym{The objective of this subsection is to evaluate the scalability and real-world applicability of Algorithm~\ref{alg:Algorithm} on a large-scale operational distribution network. Specifically, we aim to assess both the convergence performance and the algorithm’s robustness when applied to realistic grid topologies and load profiles provided by an actual utility.}
	
	\subsubsection{Experiment}
	\ym{We test Algorithm~\ref{alg:Algorithm} on a real distribution system operated by SRD, comprising $810$ buses, $8$ substations, and $1$ generator. The total load is $5.1267$ MW and $0.51$ MVAr, with a nominal voltage magnitude of $V_0=20$ kV. The network topology is illustrated in Figure~\ref{fig:SRD}. To evaluate convergence behavior, we vary the parameter $\delta$ and compute the average optimality gap across $10$ different random initializations $\boldsymbol{b}_0$. The gap is measured as:
		\begin{equation*}
			\text{Gap} = \left( \frac{\text{ADMM losses}}{\text{SRD losses}} - 1 \right) \times 100\%,
		\end{equation*}
		where the SRD losses represent reference values from the utility's solution.}
	
	\subsubsection{Results}
	\ym{Table~\ref{Table:SRD} summarizes the average gap, min-max gap, and average iteration count as a function of $\delta$. We observe a clear trade-off: increasing $\delta$ accelerates convergence but leads to a degradation in solution quality, with the optimality gap rising from $24.3\%$ to $54\%$ as $\delta$ increases from $1$ to $10^8$. Despite the model simplifications, the results demonstrate consistent convergence behavior across runs.}
	
	\begin{figure}[t!]
		\centering
		\includegraphics[width=0.9\columnwidth, height=4cm]{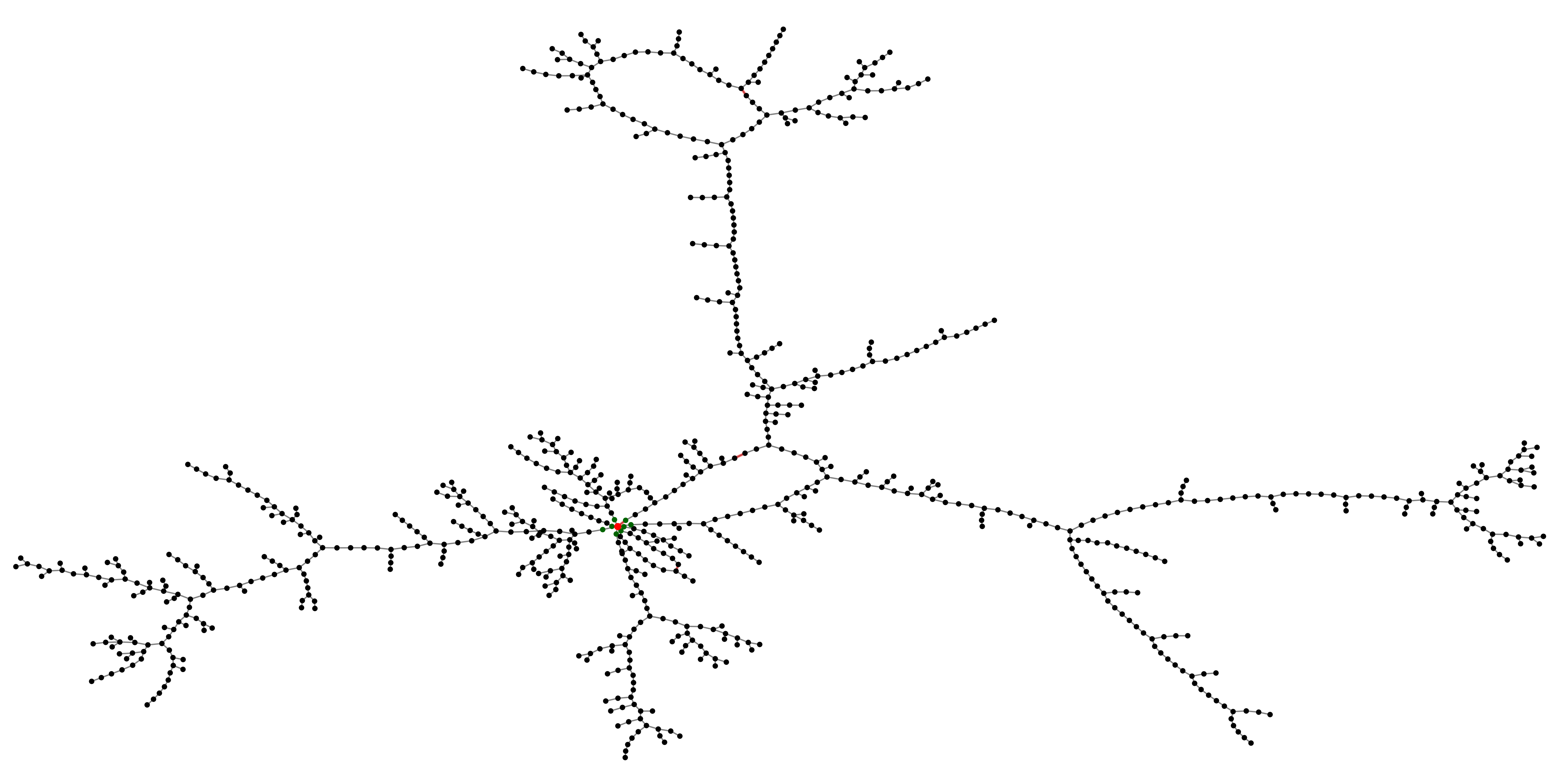}
		\caption{The SRD power grid. The big red node in the center represents the generator; green nodes indicate substations. Red lines show open switches in the optimal configuration.}
		\label{fig:SRD}
	\end{figure}
	
	\begin{table}[htb]
		\centering
		\caption{Performance of Algorithm~\ref{alg:Algorithm} on the SRD network for various $\delta$ values. Averages computed over 10 initializations.}
		\begin{tabular}{|l|c|c|c|}
			\hline
			$\delta$ & Avg. Gap (\%) & Min-Max Gap (\%) & Avg. Iterations \\ \hline
			$1$      & $24.32$        & $15.29$--$34.67$ & $19551$         \\ 
			$10$     & $34.85$        & $17.43$--$39.86$ & $15298$         \\ 
			$100$    & $39.40$        & $27.12$--$54.76$ & $12737$         \\ 
			$10^8$   & $54.02$        & $43.61$--$70.43$ & $9874$          \\ \hline
		\end{tabular}
		\label{Table:SRD}
	\end{table}
	
	\subsubsection{Resilience Evaluation}
	\ym{To assess real-time adaptability, a fault was injected at iteration $8500$ by deactivating a line, under the setting $\delta=10^8$ and $\boldsymbol{b}_0 = \boldsymbol{0}$. As illustrated in Figure~\ref{fig:SRDerror}, the algorithm responded to this disturbance by adapting the solution trajectory and converging in approximately $2 \times 10^4$ iterations. This showcases its capacity to reconfigure dynamically in large-scale settings.}
	
	\begin{figure}[t!]
		\centering
		\includegraphics[width=0.9\columnwidth, height=4cm]{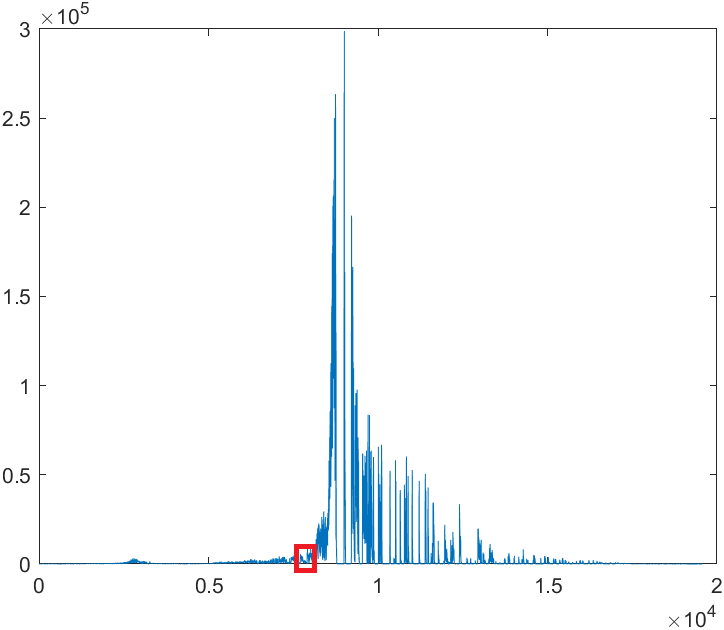}
		\caption{Evolution of $e_k$ with fault injection at iteration 8500, for $\delta = 10^8$.}
		\label{fig:SRDerror}
	\end{figure}
	
	\subsubsection{Conclusion}
	\ym{These experiments confirm the scalability, resilience, and practicality of Algorithm~\ref{alg:Algorithm} on real-world distribution networks. Despite inherent simplifications, the method achieves consistent feasibility and convergence, while adapting to dynamic changes without centralized control.}
	\section{Conclusion}
	
	This paper introduces a distributed algorithm designed to address the reconfiguration problem in PDNR, characterized by a nonlinear mixed-integer nature without straightforward solutions, especially in distributed environments. To tackle this, we leverage the Alternating Direction Method of Multipliers (ADMM), which breaks down the problem into two simpler sub-problems. Each agent is then tasked with solving a linearly constrained quadratic problem and an MWRAP with locally assigned weights. Extensive numerical experiments have thoroughly showcased the effectiveness and robustness of the proposed algorithm. These experiments provide strong evidence of its efficiency and capability to address real-world problems.

	\section*{Acknowledgments}
	
	This work was supported by the National Agency of Research (NAR) through the aLIENOR LabCom program ANR-19-LCV2-0006.
	
	\bibliographystyle{IEEEtran}
	\bibliography{biblio.bib}
	
\end{document}